\documentclass[a4paper, 12pt]{article}
\usepackage[utf8]{inputenc}
\usepackage{graphicx} 
\usepackage{float}
\usepackage{array,multirow,graphicx}
\usepackage{generalv3, fcsymbol}
\usepackage{xcolor}
\usepackage{makecell}
\usepackage[hidelinks,
            colorlinks = true,
            linkcolor = blue,
            urlcolor  = blue,
            citecolor = blue,
            anchorcolor = blueF
           ]{hyperref}

\title{Estimation with Pairwise Observations} 

\author[1]{Felix Chan}
\author[2]{László Mátyás}

\affil[1]{Curtin University, Perth, Australia}
\affil[2]{Central European University, Budapest, Hungary and Vienna, Austria}

\date{\today} 
\begin{document}
\maketitle

\vskip 2cm\noindent 
\thanks{{\bf Acknowledgement:} The authors would like to thank Tom Wansbeek for insightful discussion. Contribution by György Ruzicska in the early stages of this project and Kristóf Reizinger's invaluable coding work for the simulations are also kindly acknowledged.}
\vskip 0.5cm

\begin{abstract}
    \noindent The paper introduces a new estimation method for the standard linear regression model.  The procedure is not driven by the optimisation of any objective function rather, it is a simple weighted average of slopes from observation pairs. The paper shows that such estimator is consistent for carefully selected weights. Other properties, such as asymptotic distributions, have also been derived to facilitate valid statistical inference. Unlike traditional methods, such as Least Squares and Maximum Likelihood, among others, the estimated residual of this estimator is not by construction orthogonal to the explanatory variables of the model.  This property allows a wide range of practical applications, such as the testing of endogeneity, i.e.,the correlation between the explanatory variables and the disturbance terms.  
     
    \begin{keywords}{Linear regression model, consistent estimation, endogeneity, testing for endogeneity.}
    \end{keywords}
   \begin{JEL} {C01, C13, C20, C26, C51} \end{JEL}
\end{abstract}


\pagebreak

\section{Introduction}
To start with, the paper intents to answer a simple question relating to the estimation of a linear regression model. Let us assume that we have a very basic model with only one explanatory variable:
\begin{equation}\label{eq:dgp}
    y_i = \beta_0 + \beta_1 x_i + u_i, \qquad u_i \sim D(0, \sigma^2_u), \quad i=1,\ldots, n. 
\end{equation}
Assuming the $n$ observations the data can be divided in clusters in such a way that in each cluster there are only two observations. This clustering may be carried out in two different ways:
\begin{itemize} 
\item The $(x_{i},  y_{i})$ observations are ordered based on the values of $x_{i}$ and two adjacent observations are grouped into one cluster. Let us call this the {\it sorted adjacent} approach. There are $n-1$ such pairs. When we do not apply any sorting and just group each adjacent observations into a cluster, we call such approach {\it non-sorted adjacent}. 

\item All feasible combinations  

\begin{equation*}
    [(x_{i}, y_{i}) ; (x_{j},  y_{j})],\,\, (i \neq j)
\end{equation*}

\noindent of the observations are considered to form clusters. There are in total 
$\frac{n(n-1)}{2}$ of such clusters. Let us call this the (non-sorted) {\it full-pairwise approach}.  Here, again we may consider sorting the observations based on the values of $x_{i}$, which gives the {\it sorted full-pairwise approach}. 

\end{itemize}
In each cluster there are only two observations so the ``regression line" in a cluster is the line going through these two observations, and there is no estimation involved. In such a way we can get as many $\beta_0$ and $\beta_1$ parameters as the number of clusters (see Figures 1 and 2 for example).

\begin{figure}[ht]
\centering	
\includegraphics[scale=0.16]{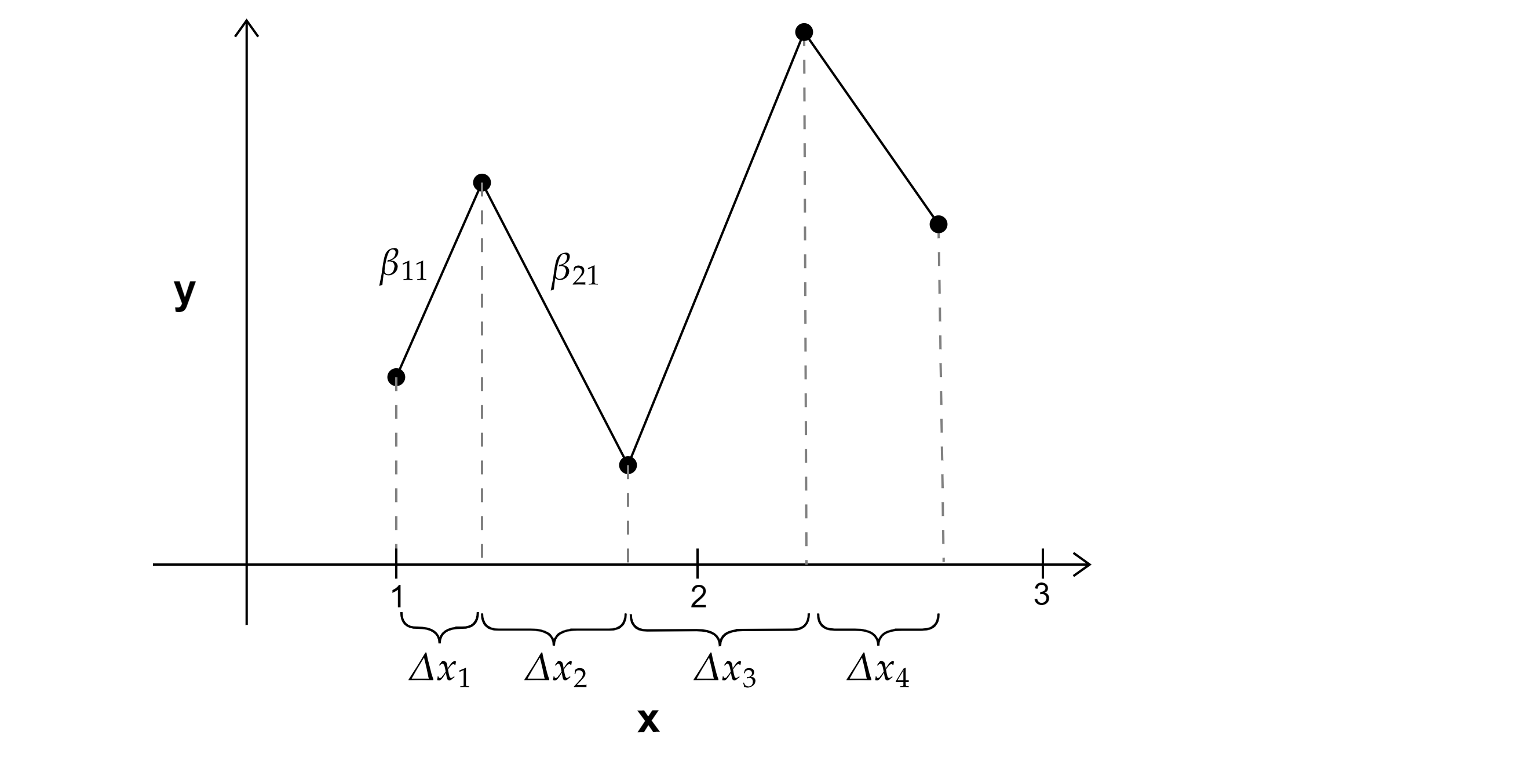}		
\caption{Sorted adjacent}
\label{Sorted}
\end{figure}

\begin{figure}[bt!]
\centering	
\includegraphics[scale=0.18]{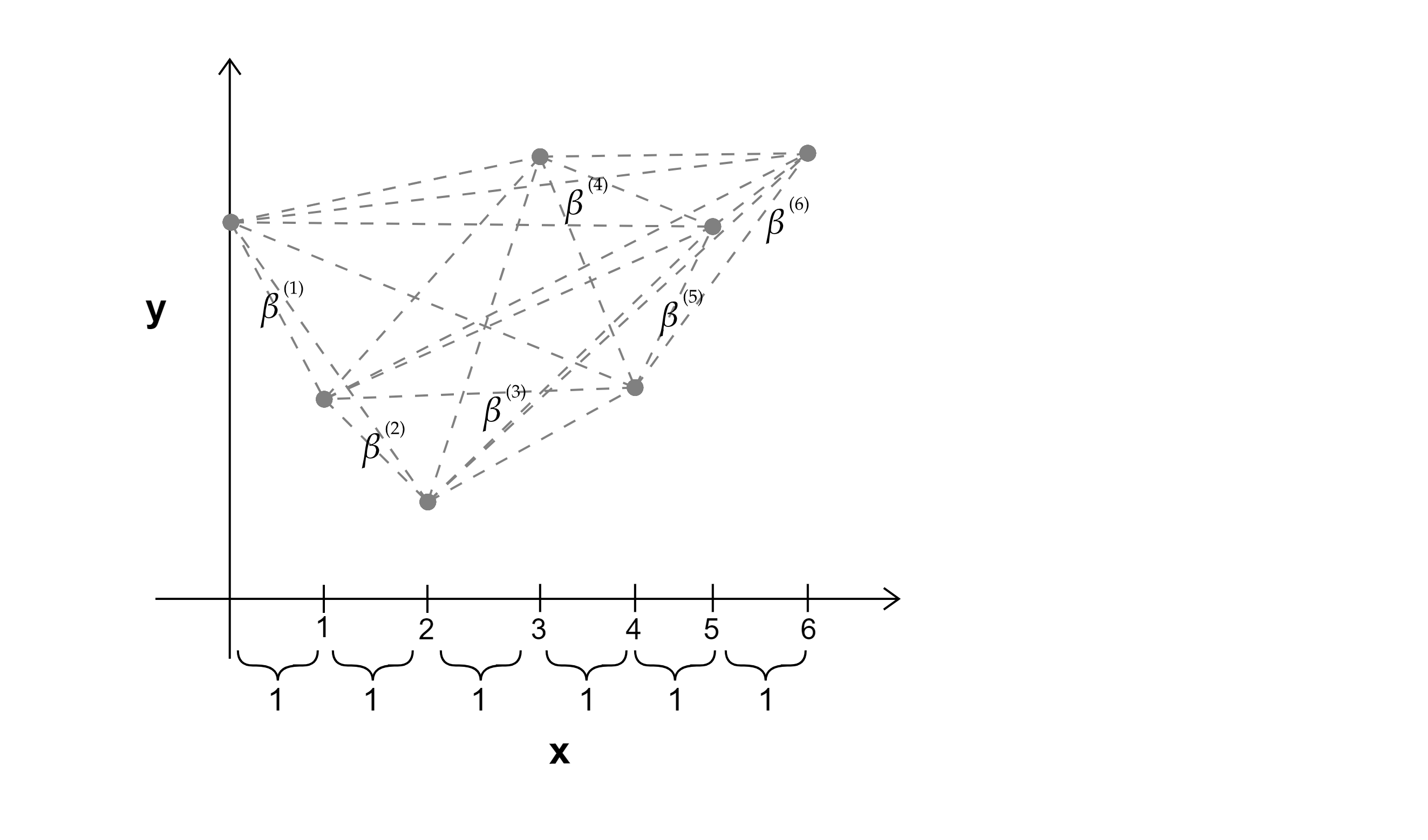}
\caption{Non-sorted full-pairwise}		
\label{Sorted2}
\end{figure}

Our main question of interest is how these ``cluster-wise" or let us call them {\it pairwise} parameters relate to the ``true" parameters of the data generation process formulated in (\ref{eq:dgp}).
Let us denote the pairwise parameters in the $i$th cluster with
 ${\beta_{i,0}}$ and  ${\beta_{i, 1}}$
and consider the following two estimators:
\begin{itemize}
\item
The weighted average of $\beta_{i,1}$ 
\begin{equation}
    \hat{\beta}_1 = \frac{\sum_{i=1}^{n} w_i \beta_{i,1}}{\sum_{i=1}^{n} w_i } \label{eqn2}
\end{equation}
with some $w_i$ weights.
\item Minimise a loss function such as
\begin{equation} 
    \hat{\beta}_1 = \underset{\beta_1}{\arg\min} \sum_{i = 1}^{n}  (w_{i}(\beta_{i,1}-\hat{\beta}_1))^2. \label{eqn3}
\end{equation}
\end{itemize}
\noindent Estimators for $\beta_0$ can also be defined similarly and they will be discussed with more details in Section \ref{subsec:intercept}. \par
The question now is how $\hat{\beta}_1$ relates to the true parameter, $\beta_1$ formulated in Equation (\ref{eq:dgp}). Let us call this estimation procedure {\it Estimation with Pairwise Observations} (EwPO).

\section{Some Monte Carlo Simulation Results}
First, we adopted a Petri-dish type approach and carried out a large number of Monte Carlo simulation. The setup and the main results can be found in Appendix \ref{app:MC-estimation}.\footnote{Additional simulation results can be found in the Online Supplement \cite{EwPO-Sup2023}.} The results provide some interesting insights in terms of the consistency of the proposed estimators under different choices of $w_i$. The EwPO turned out to be consistent in the following cases:

\begin{enumerate}
    \item EwPO as defined in Equation (\ref{eqn2}) with 
    \begin{enumerate}
        \item Sorted full-pairwise and $w_{ij} = x_i - x_j$  
        \item Non-sorted adjacent and $w_i = |x_i-x_{i-1}|$ 
        \item Sorted full-pairwise and $w_{ij} = |x_i-x_j|$
        \item Non-sorted full-pairwise and $w_{ij} = |x_i - x_j|$
        \item Sorted full-pairwise and $w_{ij} = \sqrt{(x_i-x_j)^2 + (y_i-y_j)^2}$
        \item Non-sorted full-pairwise and $w_{ij} = \sqrt{(x_i-x_j)^2 + (y_i-y_j)^2}$
    \end{enumerate}
    \item Loss function defined in Equation (\ref{eqn3}) 
    \begin{enumerate}
        \item Non-sorted adjacent and $w_i = x_i - x_{i-1}$   
        \item Sorted full-pairwise and $w_{ij} = x_i-x_j$
        \item Non-sorted full-pairwise and $w_{ij} = x_i-x_j$
    \end{enumerate}
Note: In case of the loss function taking the absolute value of weights does not matter as weights are squared in the functional form!
    
\end{enumerate}

\section{Some Analytical Results}
Next, let us presents some analytical explanations for the simulations results. We also present the notations and definitions used throughout this paper. 

\subsection{Notations and Definitions} 
In order to develop the theoretical foundation, let us start with introducing a few notations, definitions and concepts. Let $\Delta x_{ij}$ denotes the difference between $x_i$ and $x_j$ that is,  

\begin{equation}\label{eq:delta_f}
    \Delta x_{ij} = x_i - x_j, \qquad i,j=1,\ldots, n.
\end{equation}

\noindent In the event that the subscript is suppressed, then $j=i-1$. That is $\Delta x_i = x_i - x_{i-1}$. \par

\indent The \textit{pairwise} parameter is defined to be 
    \begin{equation}\label{eq:pairwise}
	\beta_{1,(i,j)} = \frac{\Delta y_{ij}}{\Delta x_{ij}}.
    \end{equation}
\indent In the event that $j=i-1$ we simplify the notation by writing $\beta_{1,i}$. That is $\beta_{1,i}=\beta_{1,(i,j)}$. \par

\indent It is useful to also note that 
\begin{equation}\label{eq:firstDiff}
    \Delta y_{ij} = \beta_1 \Delta x_{ij} + \Delta u_{ij} . 
\end{equation} 

\subsection{Partial Sum} 
Some of the expressions for analysing the full-pairwise case contain several partial sums which are asymptotically related to a \textit{Brownian Motion} or \textit{Wiener Process}. So it is useful to introduce some additional notations and state the relevant results here. Define $\floor{x}$ as the largest integer that is less than or equal to $x$ and $\ceil{x}$ as the smallest integer that is larger than or equal to $x$. Let $r\in [0,1]$, $u_i \sim N(0,\sigma^2_u)$ and define 
\begin{equation}\label{eq:wiener}
    W_n(r) = \frac{1}{\sqrt{n}} \sum^{\floor{rn}}_{i=1} \frac{u_i}{\sigma_u} ,
\end{equation} 
\noindent then under the conditional stated in \cite[Theorem 8.2, P.90]{billingsley:1999}, the \textit{Functional Central Limit Theorem} ensures $W_n(r) \overset{d}{\rightarrow} W(r)$, where $W(r)$ denotes a Wiener Process. 

\subsection{Adjacent Case}
 The estimator considered in this section is 
\begin{equation}\label{eq:weightedPW-adjacent}
    \hat{\beta}_1 = \left ( \sum^n_{i=2} w_i \right )^{-1} \left ( \sum^n_{i=2} w_i \beta_{1,i}\right )
\end{equation}
\noindent where $w_i$ represents the weight for the $i^{th}$ pairwise parameter. 
\subsubsection{Case 1: \texorpdfstring{$w_i = \Delta x_i$}{weight equals to Delta x}} 
\indent In this case, Equation (\ref{eq:weightedPW-adjacent}) reduces to 
\begin{align*} 
    \hat{\beta}_1 =& \left ( \sum^n_{i=2} w_i \right )^{-1} \left ( \sum^n_{i=2} w_i \beta_{1,i}\right ) \\
		   =& \left ( \sum^n_{i=2} \Delta x_i \right )^{-1} \left ( \sum^n_{i=2} \Delta y_i \right ) \\
		   =& \left ( \sum^n_{i=2} \Delta x_i \right )^{-1} \left ( \sum^n_{i=2} \beta_1 \Delta x_i + \Delta u_i \right ) \\
		    =&  \beta_1 + \frac{\sum^n_{i=2} \Delta u_i}{\sum^n_{i=2} \Delta x_i} .
\end{align*}
\noindent Now, in the case of adjacent pairwise parameter 
\begin{align*}
    \sum^n_{i=2} \Delta x_i =& (x_2 - x_1) + (x_3 - x_2) + (x_4-x_3) + \ldots + (x_n-x_{n-1}) \\
			    =& x_n-x_1.
\end{align*}
\noindent The last line follows because the first term in every bracket is being cancelled out by the last term in the next bracket. The only two exceptions are $x_1$ and $x_n$. So this gives
\begin{equation}\label{eq:pwEstimator_dx}
    \hat{\beta}_1 = \beta_1 + \frac{u_n - u_1}{x_n - x_1} .
\end{equation} 

\noindent Equation (\ref{eq:pwEstimator_dx}) holds regardless whether the data was sorted or not. The only difference is that, in the sorted case $x_n = \max \{x_i\}^n_{i=1}$ and $x_1 = \min \{ x_i \}^n_{i=1}$. Equation (\ref{eq:pwEstimator_dx}) does seem to suggest that $\hat{\beta}_1$ is not going to be consistent in this case, as $\displaystyle = \frac{\Delta u_{n1}}{\Delta x_{n1}}$  remains even when $n\rightarrow \infty$. \par

\subsubsection{Case 2: \texorpdfstring{$w_i = |\Delta x_i|$}{Weight equals to absolute Delta x}}
\indent In this case, it would be useful to first define the $\sgn{x}$ function.
\begin{equation}\label{eq:sign}
    \sgn{x} = \begin{cases} 
	    1  &\text{ if } x \geq 0 \\ 
	    -1  &\text{ if } x<0. 
    \end{cases}
\end{equation} 
\noindent Note that $\sgn{x}x = |x|$, and that unlike the case when $w_i = \Delta x_i$, 
\begin{equation*}
    \sum^n_{i=1} |\Delta x_i| = |x_2-x_1| + |x_3-x_2| + \ldots + |x_n-x_{n-1}|
\end{equation*}
\noindent and no simple cancellation occurs here when the observations are \textsl{not sorted}. However, when the observations are \textsl{sorted}, it is clear that $\sum^n_{i=2} |\Delta x_i| = \sum^n_{i=2} \Delta x_i$ since $\Delta x_i = |\Delta x_i|$. So the result for the sorted case is the same as the case when $w_i = \Delta x_i$.  

\noindent For the non-sorted case, 

\begin{align*} 
    \hat{\beta}_1 =& \left ( \sum^n_{i=2} w_i \right )^{-1} \left ( \sum^n_{i=2} w_i \beta_{1,i}\right ) \\
		   =& \left ( \sum^n_{i=2} |\Delta x_i| \right )^{-1} \left ( \sum^n_{i=2} \sgn{\Delta x_i} \Delta y_i \right ) \\
		   =& \left ( \sum^n_{i=2} |\Delta x_i| \right )^{-1} \left ( \sum^n_{i=2} \beta_1 \sgn{\Delta x_i} \Delta x_i + \sgn{\Delta x_i}\Delta u_i \right ) \\
		   =& \left ( \sum^n_{i=2} |\Delta x_i| \right )^{-1} \left ( \sum^n_{i=2} \beta_1 |\Delta x_i| + \sgn{\Delta x_i}\Delta u_i \right ) \quad \because \sgn{\Delta x_i}\Delta x_i = |\Delta x_i|\\
		    =&  \beta_1 + \frac{\sum^n_{i=2} \sgn{\Delta x_i} u_i }{\sum^n_{i=1} |\Delta x_i|} .
\end{align*}
\noindent Under the assumption that $\E \left (|\Delta x_i| \right ) \leq \infty$, it is clear that $\E \left ( |\Delta x_i| \right ) > 0$. Note that $\E (|\sgn{\Delta x_i}u_i|) = \E ( |u_i| )$, so under that assumption that $\E (u_i) = 0$, $\E (|u_i|) < \infty$ and $u_i \perp x_i$ then $\E \left [ \sgn{\Delta x_i}u_i \right ] = \E \left [ \sgn{\Delta x_i} \right ]\E [ u_i ] = 0$. Moreover, under Weak Law of Large Number (WLLN), 
\begin{equation*}
 (n-1)^{-1}\sum^n_{i=2} \sgn{\Delta x_i} u_i  = o_p(1)
\end{equation*}
\noindent and hence 
\begin{equation*}
    \hat{\beta}_1 - \beta_1 = o_p(1) .
\end{equation*}

\subsection{Full-pairwise} 
\indent An important difference here is that the number of pairwise parameters is much larger. Let $N$ denotes the number of full-pairwise parameters then $N = \displaystyle \frac{n(n+1)}{2}$. It is also important to identify the exact pair from each pairwise parameter and redefine the estimator as 
\begin{equation}\label{eq:fullpw_estimator}
\hat{\beta}_1 = \left ( \sum^n_{i=2} \sum^{i-1}_{j=1} w_{ij} \right )^{-1} \left ( \sum^n_{i=2} \sum^{i-1}_{j=1} w_{ij} \beta_{1,(i,j)} \right ) 
\end{equation} 
\noindent where $\beta_{1,(i,j)}$ is defined as in Equation (\ref{eq:pairwise}). 
\subsubsection{Case 1: \texorpdfstring{$w_{ij} = \Delta x_{ij}$}{Weight equals to Delta x}}
\indent We first deal with the \textit{non-sorted} case. In this case, the estimator reduces to 
\begin{align} 
    \hat{\beta}_1 =& \left ( \sum^n_{i=2} \sum^{i-1}_{j=1} w_{ij} \right )^{-1} \left ( \sum^n_{i=2} \sum^{i-1}_{j=1} w_{ij} \beta_{1,(i,j)} \right ) \notag \\
		=& \left ( \sum^n_{i=2} \sum^{i-1}_{j=1} \Delta x_{ij} \right )^{-1} \left ( \sum^n_{i=2} \sum^{i-1}_{j=1} \Delta y_{ij} \right ) \notag \\
		=& \beta_1 + \frac{\sum^n_{i=2} \sum^{i-1}_{j=1} \Delta u_{ij}}{\sum^n_{i=2} \sum^{i-1}_{j=1} \Delta x_{ij}}. \label{eq:beta1_fpw}
\end{align}
\noindent Let us first focus on the denominator as the argument for the numerator will be similar. 
\begin{align} 
    \sum^n_{i=2} \sum^{i-1}_{j=1} \Delta x_{ij} =& \sum^n_{i=2} \sum^{i-1}_{j=1} (x_i - x_j) \notag\\
			=& \sum^n_{i=2} \sum^{i-1}_{j=1} x_i - \sum^n_{i=2} \sum^{i-1}_{j=1} x_j \notag\\
		=& \sum^n_{i=2} (i-1) x_i - \sum^n_{i=2} \sum^{i-1}_{j=1} x_j. \label{eq:messy01}
\end{align}
\noindent For reason that will become clear later, we are going to assume $x_i\sim D(\mu_x, \sigma^2_x)$ with $\sigma^2<\infty$ and $\E |x_i|<\infty$. We analyse expression (\ref{eq:messy01}) term-by-term. The first term in expression (\ref{eq:messy01}) can be rewritten as 
\begin{align*} 
    \sum^n_{i=2} (i-1) x_i =& N\mu_x + \sum^n_{i=2} (i-1)(x_i-\mu_x) \\
	    =& N\mu_x + \sum^n_{i=2} iz_i - \sum^n_{i=2} z_i \\
	    =& N\mu_x - \sum^n_{i=2} z_i + \sum^n_{i=2} iz_i \, ,
\end{align*}
\noindent where $z_i = x_i - \mu_x$. Note that $z_i$ is the de-meaned version of $x_i$ and therefore $\E(z_i)=0$. To expand further the expression above, 
\begin{align}
	    N\mu_x - \sum^n_{i=2} z_i + \sum^n_{i=2} iz_i =& N\mu_x - \sum^n_{i=2} z_i + \sum^n_{i=1} \sum^n_{j=i} z_j \notag \\
					    =&N\mu_x - \sum^n_{i=2} z_i + \sum^n_{i=1} \sum^n_{j=\ceil{rN}} z_j \text{ where }r=i/N \notag\\
					    =&N\mu_x - \sum^n_{i=2} z_i + \sqrt{n} \sigma_x \sum^n_{i=1} \left (W_n(1) - W_n(r) \right ) \notag\\
					    =&N\mu_x - \sum^n_{i=2} z_i + n\sqrt{n} \sigma_x \sum^n_{i=1} \left ( \int^{i/n}_{(i-1)/n} W_n(1) - W_n(s)  ds\right ) \notag\\
					    =&N\mu_x - \sum^n_{i=2} z_i + n\sqrt{n} \sigma_x  \left ( W_n(1) - \int^1_0 W_n(s)  ds\right ). \label{eq:messy01_01}
\end{align}	
\noindent Now for the second term in expression (\ref{eq:messy01}), 
\begin{align}
    \sum^n_{i=2} \sum^{i-1}_{j=1} x_j =& N\mu_x +  \sum^n_{i=2} \sum^{i-1}_{j=1} z_j \notag \\
					=& N\mu_x + \sqrt{n} \sigma_x \sum^n_{i=1} W_n(r) \quad r=\frac{i-1}{n} \notag \\ 
				      =& N\mu_x + n\sqrt{n} \sigma_x \sum^n_{i=2} \int^{i/n}_{(i-1)/n} W_n(s) ds \notag \\
				      =& N\mu_x + n\sqrt{n} \sigma_x \int^1_0 W_n(s) ds. \label{eq:messy01_02}.  
\end{align} 
\noindent Substitute expressions (\ref{eq:messy01_01}) and (\ref{eq:messy01_02}) to expression (\ref{eq:messy01}) yields 
\begin{equation} \label{eq:partialSum_DeltaXij}
    \sum^n_{i=2} \sum^{i-1}_{j=1} \Delta x_{ij} = n\sqrt{n} \sigma_x \left ( W(1) - 2 \int^1_0 W(s) ds \right )  -\sum^n_{i=1} z_i. 
\end{equation}	
\noindent Note that $N=O(n^2)$ so 
\begin{equation*}
    N^{-1} \sum^n_{i=2} \sum^{i-1}_{j=1} \Delta x_{ij} = o_p(1). 
\end{equation*}	
\noindent The same argument applies to $\sum^n_{i=1} \sum^{i-1}_{j=1} \Delta u_{ij}$ and since both approach 0 at the same rate, $\hat{\beta}_1$ is going to be unstable as $n\rightarrow \infty$.\par

\indent The case when the data is sorted means that $\sum^n_{i=1} \sum^{i-1}_{j=1} \Delta x_{ij}$ is not going to be 0 asymptotically since $\Delta x_{ij} > 0$ by construction. Equation (\ref{eq:beta1_fpw}) stills hold when $x_i$ is sorted such that $x_i \geq x_j$ for all $i>j$. Since $u_i$ is independent from $x_i$ this means the numerator in Equation (\ref{eq:beta1_fpw}) approaches 0 as $n\rightarrow \infty$ but the denominator either approaches a finite non-zero positive constant or diverges. In both cases, the $\hat{\beta}_1 - \beta_1 = o_p(1)$ under the Assumptions of Proposition \ref{prop:dist_x_est}, which provides the asymptotic distribution for this case. \par

\subsubsection{Case 2: \texorpdfstring{$w_{ij} = |\Delta x_{ij}|$}{Weight equals to Absolute Delta x}} 
Following similar derivation as above, it is straightforward to show that 

\begin{align}
    \hat{\beta}_{1,(i,j)} =& {\left ( \sum^n_{i=2} \sum^{i-1}_{j=1} |\Delta x_{ij} | \right )}^{-1} \left ( \sum^n_{i=2} \sum^{i-1}_{j=1} \frac{\Delta y_{ij} }{\Delta x_{ij} } |\Delta x_{ij} | \right ) \notag \\
    =&  {\left ( \sum^n_{i=2} \sum^{i-1}_{j=1} |\Delta x_{ij}| \right )}^{-1} \left ( \sum^n_{i=2} \sum^{i-1}_{j=1} \Delta y_{ij} \sgn{\Delta x_{ij}} \right ) \notag \\
    =&  {\left ( \sum^n_{i=2} \sum^{i-1}_{j=1} |\Delta x_{ij} | \right )}^{-1} \left ( \sum^n_{i=2} \sum^{i-1}_{j=1} \beta_1 |\Delta x_{ij} | + \Delta u_{ij} \sgn{\Delta x_{ij}} \right ) \notag \\
    =& {\beta_1 +  \left ( \sum^n_{i=2} \sum^{i-1}_{j=1} |\Delta x_{ij} | \right )}^{-1} \left ( \sum^n_{i=2} \sum^{i-1}_{j=1} \Delta u_{ij} \sgn{\Delta x_{ij}} \right ) \notag \\
    =& \beta_1 +  \left ( \sum^n_{i=2} \sum^{i-1}_{j=1} |\Delta x_{ij} | \right )^{-1} \left ( \sum^n_{i=2} \sum^{i-1}_{j=1} u_i \sgn{\Delta x_{ij}} + \sum^n_{i=2} \sum^{i-1}_{j=1} u_j \sgn{\Delta x_{ij}} \right ) \, . \label{eq:tempmess01}
\end{align} 

\noindent Given that $x_i\perp u_j$ for all $i,j=1,\ldots, n$ and  $n^{-1}\sum^n_{i=1} |\sgn{\Delta x_{ij}}| \leq 1$, the last two terms in Equation (\ref{eq:tempmess01}) therefore converges to 0 following a similar argument as in the case $w_i = |\Delta x_{ij}|$. Using the same argument as above, $\sum^n_{i=1}\sum^{i-1}_{j=1} |\Delta x_{ij} | = O_p(n^2)$ with a non-zero positive bound. Thus $\hat{\beta}_1 - \beta_1 = o_p(1)$. \par

 \indent An interesting observation from the theoretical discussion above is that the intercept, $\beta_0$, does not play a role in the consistency of EwPO estimator. The assumption of exogeneity, however, remains essential for the identification and consistency of $\beta_1$ as well as $\beta_0$, as shown in Section \ref{subsec:intercept} below. 

\subsubsection{Confidence Intervals} \label{subsubsect:ci} 

\indent The asymptotic distribution of $\hat{\beta}$ in the full-pairwise case with $w_{ij} = \Delta x_{ij}$ can be found in Proposition \ref{prop:dist_x}. Analytical results for the asymptotic behaviour of the full-pairwise case with $w_{ij} = |\Delta x_{ij}|$ is more difficult as the presence of absolute value creates some technical challenges. In practice, however, hypothesis testing in this case can still be conducted via the \textit{jackknife} procedure. \par

\indent Regardless on the choice of $w_{ij}$, the distribution of $\hat{\beta}$ in the full-pairwise case is likely to be non-standard and involve stochastic integrals. Therefore, the critical values is still required to be simulated even if the analytical distributions can be obtained. From a practical viewpoint, the jackknife procedure does not necessarily induce more computational cost and it is generally much simpler to conduct as the procedure is the same across different choices of $w_{ij}$. 

\indent An example of the procedure to estimate critical values can be found as follows:

\begin{enumerate}[{Step} 1.]
    \item Set the level of significant for the test $\alpha$.
    \item Set the number of observations to be removed from the main sample, call it $d$ where $\sqrt{n} < d < n$. 
    \item Set the number of replication $R < \begin{pmatrix} n \\ n-d \end{pmatrix}$
    \item set $i=1$
    \item Create a jackknife sample, $J_i$, by removing $d$ randomly selected observations from the main sample. \label{step:sample}
    \item Compute $\hat{\beta}$ using the jackknife sample $J_i$, call it $\hat{\beta}_i$
    \item Set $i\leftarrow i+1$. \label{step:increment}
    \item Repeat Steps \ref{step:sample} to \ref{step:increment} until $i=R$. 
    \item Sort the set $\{\hat{\beta}_i\}_{i=1}^R$ from lowest to highest. 
    \item The $\displaystyle \floor{\frac{\alpha}{2} R}^{th}$ and $\floor{\frac{(1-\alpha)}{2}R}^{th}$ elements in the set gives the lower and upper bounds of the confidence interval at the $\alpha$ significant level. 
\end{enumerate}

\noindent Table \ref{table:n_abs_jackknife} provides an example on the performance of the jackknife algorithm. As shown in the table, the algorithm performs well in obtaining the confidence intervals of the full-pairwise estimator. \par

\begin{table}[H]
\centering
\begin{tabular}{|c|c|c|c|c|}
\hline
\multicolumn{1}{|c|}{Full-pairwise MC} & \multicolumn{4}{c|}{Jackknife confidence intervals lower and upper bounds}
\\ \hline
  & Parameter & \makecell{Full-pairwise,\\one single sample\\MC estimates} &\makecell{Lower bound\\($\beta_{0.025}$)}& \makecell{Upper bound\\ ($\beta_{0.975}$)} \\
\cline{2-5} 
\multirow{4}{*}{{n=50}}&  Exogen &  0.4479 & 0.3142  & 0.5618 \\
\cline{3-5} 
& $\rho=0.2$ & 0.6085 & 0.4763  & 0.8134  \\ 
\cline{3-5} 
& $\rho=0.5$ & 0.5925 & 0.4515   & 0.7357\\ 
\cline{3-5} 
& $\rho=0.8$  & 0.5971 & 0.4415   & 0.7929 \\ 
\cline{3-5} 
\hline

\multirow{4}{*}{{n=500}} &  Exogen & 0.4725 & 0.4541  & 0.5352 \\
\cline{3-5} 
& $\rho=0.2$ & 0.5347 & 0.4994  & 0.5809  \\ 
\cline{3-5} 
& $\rho=0.5$ & 0.6051 & 0.5881   & 0.6683\\ 
\cline{3-5} 
& $\rho=0.8$  & 0.6516 & 0.6188   & 0.6971 \\ 
\cline{3-5} 
\hline

\multirow{4}{*}{{n=1000}} &  Exogen & 0.5027 & 0.4388  & 0.5003 \\
\cline{3-5} 
& $\rho=0.2$ & 0.5558 &  0.5106  & 0.5660  \\ 
\cline{3-5} 
& $\rho=0.5$ &0.6101 & 0.5626   & 0.6195\\ 
\cline{3-5} 
& $\rho=0.8$  & 0.6527 &  0.6404   & 0.6912 \\ 
\cline{3-5} 
\hline

\multirow{4}{*}{{n=5000}} &  Exogen  & 0.5117 &0.4899  & 0.5155 \\
\cline{3-5} 
& $\rho=0.2$ & 0.5415 & 0.5291  & 0.5554  \\ 
\cline{3-5} 
& $\rho=0.5$ & 0.6007 & 0.5904   & 0.6155\\ 
\cline{3-5} 
& $\rho=0.8$  & 0.6644 & 0.6561   & 0.6801 \\ 
\cline{3-5} 
\hline
\end{tabular}

\caption{Jakckknife simulation of $\beta$ quantiles for one given sample of the Monte Carlo (MC) simulation results which are reported in the $2^{\text{nd}}$ column,
DGP {$\sim$} N(0,5), $|\Delta x|$ weighted full-pairwise estimator, number of repetitions: 10000 and jackknife sample size ($d$): 2535 }
\label{table:n_abs_jackknife}
\end{table}

\subsection{What about the Intercept?} \label{subsec:intercept} 
The discussion above focuses on the slope. This section investigates the relation between pairwise parameter and the intercept in Equation (\ref{eq:dgp}). The most obvious approach is to consider 
\begin{equation} \label{eq:naive_beta0}
    \hat{\beta}_0 = \bar{y}_n - \hat{\beta}_1 \bar{x}_n \,,
\end{equation} 
\noindent where $\hat{\beta}_1$ is any consistent estimator of $\beta_1$ with $\bar{y}_n$ and $\bar{x}_n$ denoting the averages of $y_i$ and $x_i$ over a sample of $n$ observations, respectively. That is, $\bar{y}_n = n^{-1} \sum^n_{i=1} y_i$ and similarly for $\bar{x}_n$.\par

\indent It is straightforward to show that $\hat{\beta}_0$ is consistent since 
\begin{align*}
    \hat{\beta}_0 =& \bar{y}_n - \hat{\beta}_1 \bar{x}_n \\
		=& \left (\beta_1 - \hat{\beta}_1 \right ) \bar{x}_n + \bar{u}_n
\end{align*}	
\noindent and as $n\rightarrow \infty$, $\beta_1 - \hat{\beta}_1 = o_p(1)$, $\bar{x}_n - \E(x_i) = o_p(1)$ an $\bar{u}_n = o_p(1)$. Thus, by the Continuous Mapping Theorem, $\hat{\beta}_0 - \beta_0 = o_p(1)$. \par

\indent While this provides a consistent estimator, we would also like to examine if there is a connection between the pairwise parameter of $\beta_0$ and $\beta_0$. Define the pairwise parameter of the intercept as 
\begin{equation}\label{eq:intercept_pw}
    \beta_{0, (i,j)} = y_i - \beta_{1, (i,j)} x_i. 
\end{equation}	
\subsubsection{The adjacent case}
\noindent In the adjacent case we consider the following estimator of $\beta_0$
\begin{equation}
    \hat{\beta}_0 = {\left ( \sum^n_{i=1} w_i \right )}^{-1} \left ( \sum^n_{i=1} w_i \beta_{0,i} \right ). 
\end{equation}	
\noindent Since with $w_i= |\Delta x_i|$  the non-sorted observations is the only case that produces consistent estimate of $\beta_1$, we will focus on this particular case for $\beta_0$. 
\begin{align}
    \hat{\beta}_0 =& {\left ( \sum^n_{i=2} |\Delta x_i | \right )}^{-1} \left ( \sum^n_{i=2} |\Delta x_i| \beta_{0,i} \right ) \notag\\
    =& \beta_0 + {\left ( \sum^n_{i=2} |\Delta x_i| \right ) }^{-1} \left [ \beta_1\sum^n_{i=2} |\Delta x_i| x_i  -  \sum^n_{i=2} |\Delta x_i | x_i \beta_{1,i} + \sum^n_{i=2} |\Delta x_i|u_i \right ]. \label{eq:intercept_intermediate_absx}
\end{align}	
\noindent Expanding the following term further gives
\begin{align*}
    \sum^n_{i=2} |\Delta x_i| x_i \beta_{1,i} =& \sum^n_{i=2} \sgn{\Delta x_i}x_i \Delta y_i   \\
		=& \beta_1 \sum^n_{i=2} |\Delta x_i| x_i + \sum^n_{i=2} \sgn{\Delta x_i} x_i\Delta u_i.  
\end{align*}	
\noindent Substitute the last expression above into Equation (\ref{eq:intercept_intermediate_absx}) yields
\begin{equation}
    \hat{\beta}_0 = \beta_0 + {\left ( \sum^n_{i=2} |\Delta x_i| \right )}^{-1} \left ( \sum^n_{i=2} |\Delta x_i| u_i  - \sum^n_{i=2} \sgn{\Delta x_i} x_i\Delta u_i \right ). 
\end{equation}	
\noindent Under the assumption that  $x_i \perp u_i$ and the usual regularity assumptions on $x_i$, then by the WLLN, $\sum^n_{i=2} |\Delta x_i| u_i = o_p(1)$ and $\sum^n_{i=2} \sgn{\Delta x_i} x_i \Delta u_i$. The former holds, because the sum is converging to $\E \left ( |\Delta x_i| u_i \right ) = \E \left (|\Delta x_i| \right ) \E (u_i) = 0$ and the second expression holds because the sum is converging to $\E \left [ \sgn{\Delta x_i} x_i \Delta u_i \right ] = \E \left [\sgn{\Delta x_i} x_i\right ] \E \left (\Delta u_i \right )=0$. \par

\subsection{Loss Function}
Next, let us analyse the properties of estimator (\ref{eqn3}):

\begin{equation} \label{eq:quadraticLoss}
    \hat{\beta} = \underset{\beta}{\arg\min} \sum^n_{i=1} \left [ w_i \left ( \beta_i - \beta\right ) \right ]^2 \,,
\end{equation}
\noindent where 
\begin{equation*} 
    \beta_i = \frac{\Delta y_i}{\Delta x_i}.
\end{equation*}
\subsubsection{\texorpdfstring{$w_i = |\Delta x_i|$}{Weight equals to Absolute Delta x}}
Given the quadratic loss, here show that the case when $w_i = \Delta x_i$ is identical to the case when $w_i = |\Delta x_i|$. \par

The solution to Equation (\ref{eq:quadraticLoss}) implies 
\begin{equation*}
    \sum^n_{i=1} w_i^2 \left ( \beta_i - \hat{\beta} \right ) = 0,  
\end{equation*}
\noindent and therefore
\begin{equation*}
   \hat{\beta} = \left (\sum^n_{i=1} w_i^2 \right )^{-1} \left ( \sum^n_{i=1} w_i^2 \beta_i \right ).
\end{equation*}
Thus when $w_i = |\Delta x_i|$ or when $w_i = \Delta x_i$ we get
\begin{equation*}
   \hat{\beta} = \left (\sum^n_{i=1} |\Delta x_i|^2 \right )^{-1} \left ( \sum^n_{i=1} \Delta x_i \Delta y_i \right ) \,,
\end{equation*}
\noindent which is in fact the OLS estimator for the `first difference' equation
\begin{equation}
    \Delta y_i = \Delta x_i\beta  + \Delta u_i \,.
\end{equation}

\subsubsection{Full-pairwise} 
Reconsider the estimator in the full-pairwise case

\begin{equation}
    \hat{\beta} = \underset{\beta}{\arg\min} \sum^n_{i=1} \sum^{i-1}_{j=1} \left [ w_{ij} \left ( \beta_{ij} - \beta\right ) \right ]^2 \,,
\end{equation}

\noindent which gives 
\begin{equation*}
    \hat{\beta} = \left (\sum^n_{i=1} \sum^{i-1}_{j=1} w_{ij}^2 \right )^{-1} \left ( \sum^n_{i=1}\sum^{i-1}_{j=1} w_{ij}^2 \Delta \beta_{ij} \right ).
\end{equation*}

\noindent Now, if we set $w_{ij}=|\Delta x_{ij}|$ then 
\begin{equation} \label{eq:betahatInBeta}
    \hat{\beta} = \beta +  \left (\sum^n_{i=1} \sum^{i-1}_{j=1} |\Delta x_{ij}|^2 \right )^{-1} \left (  \sum^n_{i=1} \sum^{i-1}_{j=1} \Delta x_{ij}\Delta u_{ij} \right ) \,,
\end{equation}
\noindent which under the assumption that $x_i\perp u_i$, shows that the estimator is consistent. The same applies to $w_{ij} = \Delta x_{ij}$. The equivalence is mostly due to the quadratic loss.  \par 
\indent An interesting case is $w_{ij} = \sqrt{|\Delta x_{ij} |}$, then 
\begin{equation*}
    \hat{\beta} =  \left (\sum^n_{i=1} \sum^{i-1}_{j=1} |\Delta x_{ij}| \right )^{-1}  \left (\sum^n_{i=1} \sum^{i-1}_{j=1} |\Delta x_{ij} | \hat{\beta}_{ij} \right ) \,,
\end{equation*}
\noindent which is in fact our full-pairwise estimator. \par

\section{Multivariate Extension} \label{sec:multivariate}

One way to extend the approach to multiple explanatory variables is to utilise the residual matrix. Consider the following Data Generating Process

\begin{equation} \label{eq:dgp_multi}
    y_i = \beta_0 + \sum^K_{k=1} \beta_k x_{ki} + u_i \qquad u_i\sim D(0, \sigma^2_u) \quad i=1,\ldots, n 
\end{equation}	
\noindent with the matrix representation 

\begin{equation} \label{eq:dgp_multi_matrix}
    \bfy = \bfX\beta + u\,,
\end{equation}	
\noindent where $\bfy = \left (y_1, \ldots, y_n\right )'$ is a $n\times 1$ vector containing the dependent variable, $\bfX = \left (\bfx_1, \ldots, \bfx_K \right )$ is a $n\times K$ matrix containing the $K$ explanatory variables such that $\bfX$ is $\bfx_k = \left ( x_{k1}, \ldots, x_{kn} \right )'$ for $k=1,\ldots, K$ and $\bfu = \left (u_1, \ldots, u_K \right )'$ is a $n\times 1$ vector containing the residuals. \par

It is convenient to introduce some special matrices here. These matrices allow us to represent the pairwise parameter as a sequence of matrix operations. Define $\bfd_{i,j}$ as a $1\times n$ vector such that the $i^{th}$ and $j^{th}$ elements are 1 and -1, respectively while all other elements are 0. Let $\bfD_F = \left ( \bfd_{n,n-1}', \bfd_{n,n-2}', \ldots, \bfd_{n,1}', \bfd_{n-1,n-2}', \ldots, \bfd_{n-1,1}', \ldots, \bfd_{2,1}' \right )'$ and $\bfD_a = \left ( \bfd_{2,1}', \bfd_{3,2}', \bfd_{4,3}'\ldots, \bfd_{n,n-1}'\right )'$. Given these matrices, it is straightforward to show that $\bfD_F \bfx = \{ \Delta x_{ij} \}_{i=2,j=1}^{n,i-1}$ is a $n(n-1)/2 \times 1$ vector and $\bfD_a\bfx = \{ \Delta x_i \}_{i=2}^{n}$ is a $(n-1) \times 1$ vector. The former gives the vector of Full-pairwise Difference while the latter yields the Adjacent Difference. \par

Let $\bfS:\R^n\rightarrow \R^n$ be a function that returns a $n\times 1$ zero-one selection matrix such that if $\bfz = \bfS(\bfx)\bfx$ then $\bfz$ is a $n\times 1$ vector containing the same elements as $\bfx$ such that $z_i \geq z_j$ for $j>i$. Thus, $\bfS(\bfx)\bfx$ produces the sorted version of $\bfx$ in descending order. Finally, define the residual maker 
\begin{equation}
    \bfM_k = \bfI_{n} - \bfX_{-k} {\left ( \bfX_{-k}'\bfX_{-k} \right )}^{-1}\bfX_{-k} \,,
\end{equation}	
\noindent where $\bfX_{-k}$ denote the $n\times K-1$ matrix which contains all the columns of $\bfX$ except the $k^{th}$ column. \par

The first step is to develop the matrix representation of the univariate case, then the multivariate case can be derived by repeat applications of the residual maker. That is, we  construct the multivariate version by combining the univariate estimator for each of the regressors. 

\begin{equation*}
    \bfy = \bfone\beta_0 + \bfx \bfbeta_1 + \bfu 
\end{equation*}	

\noindent and note that $\bfD \bfy = \bfD\bfx \beta_1 + \bfD\bfu$ for $\bfD = \{\bfD_a, \bfD_F\}$. The pairwise parameters for the sample can be written as 
\begin{equation} \label{eq:pairwise_multi}
    \bfbeta_1 = \bfdiag{\bfx}^{-1}  \bfD\bfy \,,
\end{equation}	
\noindent where $\bfdiag{\bfx}= {\left [ \bfI \otimes \left (\bfD \bfx \right )' \right ]\bfS_n }$ is a diagonal matrix with the elements of $\bfx$ in the diagonal. $\bfS_n$ is a $n^2\times n$ zero-one matrix such that all the $( (i-1)n+1, i)$ elements are 1 for $i=1,\ldots, n$ and all the other elements are 0. Note that Equation (\ref{eq:pairwise_multi}) is reasonably general and covers all the previously discussed case by defining the transformation matrix $\bfD$ differently. A summary can be found in Table \ref{tab:cases}.

\begin{table}[h]
    \begin{center}
    \begin{tabular}{|c|c|}
    \hline
	Cases & $\bfD$ \\ \hline
	Adjacent & $\bfD_a$ \\
	Adjacent Sorted & $\bfD_a \bfS (\bfx)$ \\ 
	Full-pairwise & $\bfD_F$ \\
	Full-pairwise Sorted & $\bfD_F \bfS (\bfx)$ \\ \hline
    \end{tabular}
    \caption{Definitions of $\bfD$ based on different construction of pairwise parameters} \label{tab:cases} 
    \end{center}
\end{table}

Given the vector of pairwise parameters $\bfbeta_1$, the estimator of $\beta_1$ can be written as 
\begin{equation} \label{eq:pairwise_estimator}
    \hat{\beta}_1 = {\left (\bfw_1 \bfi' \right )}^{-1}\bfw_1'\bfbeta_1 \,,
\end{equation}
\noindent where $\bfw_1$ is the vector of weights associated with each $\beta_{1,(i,j)}$ in $\bfbeta_1$. Following from the analysis above, some of the choices include $\bfw_1 = \left ( |\Delta x_1|, \ldots, |\Delta x_n| \right )'$ and $\bfw_1 = \left ( \Delta x_1, \ldots, \Delta x_n \right )'$. \par
\indent It is straightforward to generalise Equation (\ref{eq:pairwise_multi}), and subsequently Equation (\ref{eq:pairwise_estimator}), to the case when $\bfX$ is a $n\times K$ matrix with $K$ regressors. Consider  
\begin{align}
    \bfy =& \bfX\bfbeta + \bfu \\
    =& \bfx_k \beta_k + \bfX_{-k} \bfbeta_{-k} + \bfu \\
    \bfM_k \bfy =& \bfM_k \bfx_k \beta_k + \bfu \,, \label{eq:sequential}
\end{align}
\noindent where $\bfbeta_{-k}$ contains the same elements as $\bfbeta$ with the $k^{th}$ element removed. Following Equation (\ref{eq:pairwise_multi}) the pairwise parameter of $\beta_k$ is  
\begin{equation} \label{eq:pairwise_multi_all} 
    \bfbeta_k = \bfdiag{\bfM_k\bfx_k}^{-1} \bfD \bfM_k\bfy.
\end{equation}
\noindent Similarly, Equation (\ref{eq:pairwise_estimator}) can be generalised as
\begin{equation} \label{eq:pairwise_estimator_full}
    \hat{\beta}_k ={ \left ( \bfw_k \bfi' \right ) }^{-1} \bfw_k'\bfbeta_k
\end{equation}
\noindent and 
\begin{equation}
    \hat{\bfbeta} = \begin{bmatrix} { \left ( \bfw_1 \bfi' \right ) }^{-1} \bfw_1' &\bfzero & \ldots & \bfzero & \bfzero \\ 
					\bfzero & \ddots & \vdots & \ldots &\bfzero \\
					\bfzero & \ldots &  { \left ( \bfw_k \bfi' \right ) }^{-1} \bfw_k' & \ldots & \bfzero \\
					\vdots & \ldots &\vdots & \ddots & \ldots \\
					\bfzero & \ldots & \ldots & \ldots & { \left ( \bfw_K \bfi' \right ) }^{-1} \bfw_K' 
	\end{bmatrix} \begin{bmatrix} \bfbeta_1 \\ \vdots \\ \bfbeta_k \\ \vdots \\ \bfbeta_K \end{bmatrix}.
\end{equation}
    \noindent As shown in previous results, suitable choice of $\bfw_k$ leads to consistency of $\hat{\beta}_k$ and therefore $\hat{\bfbeta} - \bfbeta = o_p(1)$ under the same conditions as those discussed above for all $k=1,\ldots, K$. 

\section{Testing for Endogeneity} 
There are multiple endogeneity tests available for empirical use, but all rely on some kind of additional or external information, mostly in the form of instrumental variables
(see e.g., \cite{Hausmann1978}, or \cite{Wooldridge2002}, pp.~118-122).

The aim in this section is to develop two tests for endogeneity based on the estimator of the form 

\begin{equation*}
    \hat{\beta} = \left ( \sum^n_{i=2} \sum^{i-1}_{j=1} w_{ij} \right )^{-1}  \left ( \sum^n_{i=2} \sum^{i-1}_{j=1} w_{ij} \beta_{ij} \right )
\end{equation*}
which solely relies on data already in use for the estimation. \par

\subsection{Residuals Test}
The residuals test is particularly useful when $\beta_0=0$. Consider $w_{ij} = |\Delta x_{ij}|$ which implies

\begin{equation} \label{eq:estimator}
    \hat{\beta} = \beta + \delta_n  \,,
\end{equation}
\noindent where 
\begin{equation} \label{eq:delta}
    \delta_n = \left ( \sum^n_{i=2} \sum^{i-1}_{j=1} |\Delta x_{ij}| \right )^{-1} \left ( \sum^n_{i=2} \sum^{i-1}_{j=1} \sgn{\Delta x_{ij}} \Delta u_{ij} \right )
\end{equation}
\noindent and under $x_i \perp u_i$, $\delta_n = o_p(1)$. Now given the following specification
\begin{equation} \label{eq:lm_nointercept}
    y_i = x_i\beta + u_i 
\end{equation}
\noindent the estimated residual from the EwPO estimator is 
\begin{align*}
    \hat{u}_i =& y_i - x_i\hat{\beta} \\
    &= x_i (\beta - \hat{\beta}) + u_i \\
    &= -x_i \delta_n + u_i.
\end{align*}
It is straightforward to show that
\begin{align*}
    n^{-1} \sum^n_{i=1} \hat{u}_i =& -n^{-1}\sum^n_{i=1} x_i\delta_n + u_i \\
    =& -\delta_n n^{-1}\sum^n_{i=1} x_i + n^{-1} \sum^n_{i=1} u_i \,.
\end{align*}
\noindent Now as $n\rightarrow \infty$ the last line above is 
\begin{equation} \label{eq:bias}
    n^{-1} \sum^n_{i=1} \hat{u}_i = -\delta_n \mu_x + o_p(1)
\end{equation}
\noindent under the assumption that $\E(u_i) = 0$. There are two cases when $n^{-1}\sum^n_{i=1} \hat{u}_i$ is $o_p(1)$ namely, $x_i \perp u_i$ or $\mu_x = 0$.  Assuming $\mu_x\neq 0$, which can be easily verified in practice, it is possible to \textit{directly} test $x_i \perp u_i$ by testing $H_0: \E(u_i) = 0$. This provides the foundation for testing endogeneity by examining the mean of the estimated residuals using standard testing procedure, such as the $t$-test.\footnote{The argument here also applies to the usual Ordinary Least Squares (OLS) estimator. That is, when $\beta_0=0$, the estimated residuals do not have 0 mean. Thus, it also provides a test of endogeneity in this special case.} 

\subsection{Covariance Test}

Of course $\beta_0=0$ is often too restrictive in practice. One way to alleviate this is to remove the intercept by considering the model in difference. Consider $\Delta \hat{u}_{pq} = \Delta y_{pq} - \Delta x_{pq} \hat{\beta}$ i.e., the `estimated' residuals in the form 
\begin{equation}
    \begin{split}
        \Delta \hat{u}_{pq} =& \Delta y_{pq} - \Delta x_{pq} \hat{\beta} \\
         =& \Delta x_{pq} \left ( \beta - \hat{\beta} \right ) + \Delta u_{pq}.
    \end{split}
\end{equation}
\noindent In the case of a consistent estimator $\hat{\beta}$ can be expressed in the form of $\hat{\beta} = \beta + \delta(\bfw, \bfu)$ where $\bfw$ and $\bfu$ denote the vectors of the weights i.e., $w_{ij}$ and the residuals $u_i$, respectively, for $i,j=1,\ldots, n$ and $B$ vanishes to 0 asymptotically under ideal conditions, Therefore
\begin{equation*}
    \Delta \hat{u}_{pq} = \Delta x_{pq} \delta(\bfw, \bfu ) + \Delta u_{pq}. 
\end{equation*}
Hence, one can test $H_0:\E [\delta (\bfw, \bfu)] = 0$ by considering the test statistics 

\begin{equation} \label{eq:testStatistics}
    \begin{split}
        S(\bfw) = n^{-2} \sum^n_{p=2} \sum^{p-1}_{q=1} \Delta x_{pq} \Delta \hat{u}_{pq}.
    \end{split}
\end{equation}
\subsubsection{Case 1: \texorpdfstring{$\bfw = \{\Delta x_{ij} \}_{i,j=1}^n$}{Delta x }}
In this case, 
\begin{equation*}
    \Delta \hat{u}_{pq} = \Delta x_{pq}  \left ( \sum^n_{i=2} \sum^{i-1}_{j=1} \Delta x_{ij} \right )^{-1}  \left ( \sum^n_{i=2} \sum^{i-1}_{j=1} \Delta u_{ij} \right ) + \Delta u_{pq}
\end{equation*}
\noindent and the test statistics 

\begin{equation}
    \begin{split}
    S(\bfw) =& n^{-2} \left (\sum^n_{p=1} \sum^{p-1}_{q=1} \Delta x_{pq}^2\right ) \left ( \sum^n_{i=2} \sum^{i-1}_{j=1} \Delta x_{ij} \right )^{-1}  \left ( \sum^n_{i=2} \sum^{i-1}_{j=1} \Delta u_{ij} \right ) \\
    &\qquad + n^{-2} \left ( \sum^n_{p=2} \sum^{p-1}_{q=1} \Delta x_{pq} \Delta u_{pq} \right ) \,.
    \end{split}
\end{equation}
\noindent Under the null that $x_i \perp u_i$, $S(\Delta x_{ij})$ has mean 0. The asymptotic behaviour of EwPO estimator as defined in Equation (\ref{eq:fullpw_estimator}) with $w_{ij} = \Delta x_{ij}$ and that of the test statistics as defined above can be found in Propositions \ref{prop:dist_x_est} and \ref{prop:dist_x}, respectively. \par 
\indent Before presenting the theoretical results, consider the following assumptions

\begin{enumerate}[{Assumption} 1.]
    \item $\E (x_i) = \mu_x <\infty$ $\forall i=1,\ldots, n$ and $\exists \sigma^2_x < \infty$ such that $n^{-1} \sum^n_{i=1} (x_i-\mu_x)^2 - \sigma_x^2 = o_p(1)$. \label{ass:moment_x}
    \item $\E (u_i) = 0$ $\forall i=1,\ldots, n$ and $\exists \sigma^2_u < \infty$ such that $n^{-1} \sum^n_{i=1} u^2_i - \sigma^2_u = o_p(1)$.  \label{ass:moment_u}
    \item $\E (u_i| x_j) =0$ for $i,j=1,\ldots, n$. \label{ass:exogeneity}
    \item $\Delta x_{ij}$ and $\Delta u_{ij}$ are mixingales as defined in \cite{mcleish_invariance_1975} and satisfy all the conditions for Theorem 3.8 in \cite{mcleish_invariance_1975}. \label{ass:mixingale}
\end{enumerate}

\noindent Assumptions \ref{ass:moment_x} and \ref{ass:moment_u} are standard in the sense that the existence of second moments are required for both the regressor and the residuals. The exogeneity assumption as presented in Assumption \ref{ass:exogeneity} is required for the consistency of the estimator. It also stipulates the null underlying the distribution of the test statistics as asserted in Proposition \ref{prop:dist_x}. Assumption \ref{ass:mixingale} is required to ensure that the various partial sums converge to the  Brownian Motion process using the \textit{Functional Central Limit Theorem}. It is required due to the partial sums from the full-pair construction of the estimator. See also \cite{wooldridge_invariance_1988} for potentially more general assumptions on $x_i$ and $u_i$. 

\indent The asymptotic distribution of the full-pairwise EwPO with $w_{ij} = \Delta x_{ij}$ is presented in Proposition \ref{prop:dist_x_est} below

\begin{proposition} \label{prop:dist_x_est}
    Under Assumptions \ref{ass:moment_x} to \ref{ass:mixingale} with $\hat{\beta}$ as defined in Equation ({\ref{eq:fullpw_estimator}}) and $\bfw = \{\Delta x_{ij} \}^n_{i,j=1}$: 
    \begin{equation}
        \hat{\beta} - \beta \overset{d}{\rightarrow} \frac{ W(1) - 2\int^1_0 W(\lambda) d\lambda}{B(1) - 2\int^1_0 B(\lambda) d\lambda} .
    \end{equation}
\end{proposition}

\begin{proof}
    See Appendix \ref{app:proof}. 
\end{proof}

\noindent The proposition below gives the asymptotic distribution of the covariance test statistics generated by the estimator considered in Proposition \ref{prop:dist_x_est} above. 

\begin{proposition} \label{prop:dist_x}
Under the assumptions of Proposition \ref{prop:dist_x_est}

\begin{equation} \label{eq:asymptotic_dist}
\begin{split}
    (\sigma_x\sigma_u)^{-1} S(\bfw) \overset{d}{\rightarrow}& B(1)W(1) + \int^1_0 B(\lambda) W(\lambda) d\lambda - \int^1_0B(\lambda)dW(\lambda) \\&\quad - \int^1_0W(\lambda) dB(\lambda) - B^2(1) \frac{ W(1) - 2\int^1_0 W(\lambda) d\lambda}{B(1) - 2\int^1_0 B(\lambda) d\lambda}, 
\end{split}
\end{equation}
\noindent where $B(\lambda)$ and $W(\lambda)$ are independent standard Brownian motion processes. 
\end{proposition}

\begin{proof}
See Appendix \ref{app:proof} 
\end{proof}

\begin{remark}
It can be shown that both terms on the right hand side of Equation (\ref{eq:asymptotic_dist}) are negative on expectation. Under the alternative that $u_i$ and $x_i$ are correlated, the values of the last two stochastic integrals in Equation (\ref{eq:asymptotic_dist}) increase as the correlation increases. 
\end{remark} 

\begin{remark}
    The asymptotic distribution as stated in Equation (\ref{eq:asymptotic_dist}) is obviously non-standard and must be simulated in order to obtain the critical value(s) for the purpose of inference. Since $B(\lambda)$ and $W(\lambda)$ are independent Brownian processes, each of the stochastic integrals can be simulated as 
    \begin{align}
        \int^1_0 B(\lambda) d\lambda \approx& \sum^n_{p=1} z_{1p} \label{eq:stochfirst}\\
        \int^1_0 W(\lambda) d\lambda \approx& \sum^n_{p=1} z_{3p} \\
    \end{align}
\noindent where $z_{ip} = z_{ip-1} + \varepsilon_{ip}$, $\bfepsilon_p = (\varepsilon_{1p}, \varepsilon_{2p})' \sim NID(0, \bfI)$ and 
\begin{align}
    \int^1_0 W(\lambda) dB(\lambda) \approx & \sum^n_{p=1} z_{2p}\varepsilon_{1p} \\
    \int^1_0 B(\lambda) dW(\lambda) \approx & \sum^n_{p=1} z_{1p}\varepsilon_{2p} \label{eq:stochlast}. 
\end{align}
\noindent Under the null, other quantities in the distribution namely, $\sigma_x$ and $\sigma_u$ can be estimated consistently from the data and estimated residuals, respectively. Thus the asymptotic distribution can be simulated by repeated computation of Equations (\ref{eq:stochfirst}) - (\ref{eq:stochlast}) and the substitution their values into Equation (\ref{eq:asymptotic_dist}) for desired number of replications. For further details on the simulation of stochastic integrals, see for example, \cite{johansen:1995}. \par

By way of demonstration, the critical values from the simulated sample of Equation (\ref{eq:asymptotic_dist}) using Equations (\ref{eq:stochfirst}) - (\ref{eq:stochlast}) with $\sigma_x=\sigma_u = 1$ and $n=10000$ can be found in Table \ref{tab:cv} and the respective histogram in Figure \ref{fig:cv}. Note that this is simulated with $w_{ij} = \Delta x_{ij}$ and may not be appropriate when $w_{ij} = |\Delta x_{ij}|$. In this case, for example, the jackknife procedure may provide more accurate approximation.  
\end{remark}
{\color{blue}
\begin{figure}[ht]
\begin{center}
    \includegraphics[scale=0.5]{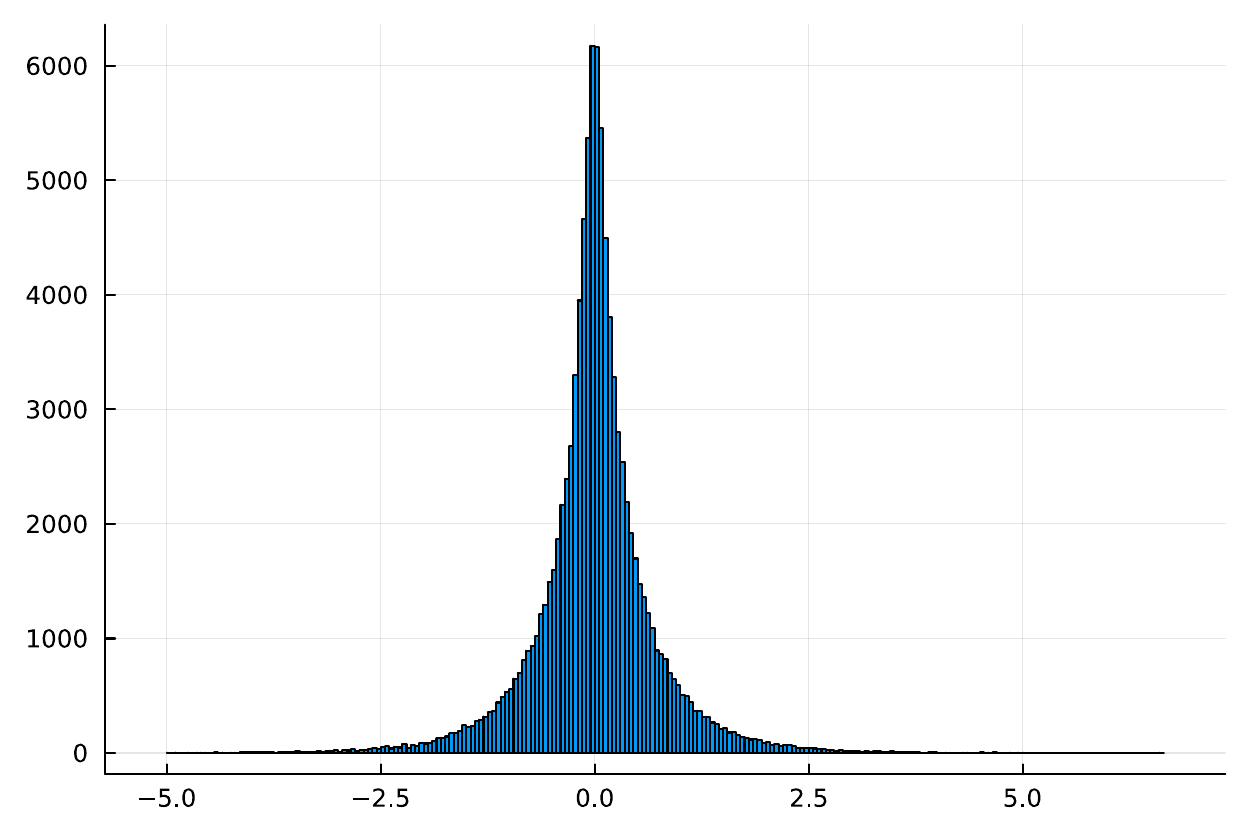}
    \caption{Simulated Sample of the Distribution} \label{fig:cv}
\end{center}
\end{figure}

\begin{table}[]
    \begin{center}
    \begin{tabular}{|c|cc|}
    \hline
         $\alpha$ & Lower CV & Upper CV   \\ \hline
         0.01 \% & -4.129 & 3.913 \\
         0.05 \% & -3.021 & 2.912 \\
         0.10 \% & -2.530 & 2.498  \\ \hline
    \end{tabular}
    \caption{Simulated Critical Values: $w_{ij} = \Delta x_{ij}$} \label{tab:cv}
    \end{center}
\end{table}
}
\subsubsection{Case 2: \texorpdfstring{$w_{ij} = |\Delta x_{ij}|$}{Absolute Delta x}}
\noindent Similar derivation as above yields
\begin{equation}
    \begin{split}
    S(\bfw) =& n^{-2} \left (\sum^n_{p=1} \sum^{p-1}_{q=1} \Delta x_{pq}^2\right ) \left ( \sum^n_{i=2} \sum^{i-1}_{j=1} |\Delta x_{ij}| \right )^{-1}  \left ( \sum^n_{i=2} \sum^{i-1}_{j=1} \sgn{\Delta x_{ij}} \Delta u_{ij} \right ) \\
    &\qquad + n^{-2} \left ( \sum^n_{p=2} \sum^{p-1}_{q=1} \Delta x_{pq} \Delta u_{pq} \right ) \,.
    \end{split}
\end{equation}


The asymptotic distribution of EwPO with $w_{ij}=|\Delta x_{ij}|$ and the associated test statistics are presented in Propositions \ref{prop:abs_diffx_dist} and \ref{prop:abs_diffx_exotest}.

\begin{proposition} \label{prop:abs_diffx_dist}
     Under Assumptions \ref{ass:moment_x} to \ref{ass:mixingale} with $\hat{\beta}$ as defined in Equation ({\ref{eq:fullpw_estimator}}) and $\bfw = \{|\Delta x_{ij}| \}^n_{i,j=1}$: 
    \begin{equation}
        \sqrt{n} \left (\hat{\beta} - \beta \right ) \overset{d}{\rightarrow} 2\frac{\sigma_u}{\mu_{\Delta x_s}} \left ( W(1) - 2\int^1_0 W(\lambda) d\lambda \right ) .
    \end{equation}
    \noindent where $\mu_{\Delta x_s} = \E (\Delta x_s)$ where $x_s$ denotes the sorted version of the regressor $x$ i.e., $\Delta x_{s,ij} = x_i - x_j$ where $x_i > x_j$.  
\end{proposition}

\begin{proof}
    See Appendix \ref{app:proof}.
\end{proof}

\begin{proposition} \label{prop:abs_diffx_exotest}
     Under Assumptions \ref{ass:moment_x} to \ref{ass:mixingale} with $\hat{\beta}$ as defined in Equation ({\ref{eq:fullpw_estimator}}) and $\bfw = \{|\Delta x_{ij}| \}^n_{i,j=1}$: 
     \begin{equation}
     \begin{split}
         (\sigma_x\sigma_u)^{-1} S(\bfw) \overset{d}{\rightarrow} B(1)W(1) +& \int^1_0 B(\lambda)W(\lambda) d\lambda\\
         & - \int^1_0 B(\lambda)dW(\lambda) - \int^1_0 B(\lambda) dW(\lambda).  
     \end{split}
     \end{equation}
\end{proposition}

\begin{proof}
    See Appendix \ref{app:proof}. 
\end{proof}

\indent Appendix \ref{app:MC-test} contains some Monte Carlo simulations of the test statistics with the two different weighting schemes under the null and selected alternatives. As shown in Appendix \ref{app:MC-test}, the mean test statistics deviates from zeros as the correlation between the regressor and the error term grows. Further investigations suggest that the covariance test has the right size with relatively good power, even in small sample, so could really be useful in practice.  

\subsection{Implications}

Basically, the results can be used in two ways. First, we could calculate the appropriate critical values for a given data set and perform the testing procedure formally. However, this requires the simulation of the critical values which is not always convenient in practice.  

Another way, however, to look at this is based on the fact that the value of the test statistic is a monotonically increasing function (in absolute value) of the correlation between $x_i$ and $u_i$, i.e., the degree of endogeneity. Therefore, the following selection procedure can be performed to reduce or minimise the risk of estimation bias due to endogeneity:

\begin{enumerate}[a)]
    \item Estimate the model, e.g., $y_i = \beta x_i +u_i$ with EwPO;
    \item Estimate the model with EwPO after being transformed by the potential instrumental variable (IV) candidate, $z_i = \beta w_i + v_i$ where $z_i = g_iy_i$, $w_i = g_ix_i$ and $v_i = g_iu_i$ with $g_i$ an IV candidate that satisfies the same assumptions as $x_i$; 
    \item Select the model for inference with the lowest statistics value. This  minimises the risk of endogeneity bias. 
\end{enumerate}
\noindent Note that the procedure above does not require the simulation of any distribution but a repeat application of the EwPO estimator. Thus, it is potentially a simpler but effective procedure.  

\subsection{Hausman Test}

A third test is inspired by the well known Hausman Test proposed in \cite{Hausmann1978}. The general idea is that the difference between OLS and EwPO with a suitably chosen weight should be zero asymptotically under the assumption of exogeneity. While EwPO may not be consistent in the presence of endogeneity, its asymptotic bias is generally not the same as that of the OLS depending on the choice of weight. This means their difference is non-zero under the alternative of endogeneity.  This can be easily demonstrated by comparing the OLS estimator with an EwPO estimator. Recall the OLS estimator for a simple linear regression model $y_i = \beta_0 + \beta_1 x_i + u_i$ can be expressed as 

\begin{align*}
    \hat{\beta}_1 =& \beta_1 + \frac{\sum^n_{i=1} (x_i-\bar{x})u_i}{\sum^n_{i=1} (x_i-\bar{x})^2} \\
    =& \beta_1 + \epsilon_n. 
\end{align*}

\noindent Let $\delta_n(w_{ij}) = \left ( \sum^n_{i=2}\sum^{i-1}_{j=1} w_{ij} \right )^{-1} \left ( \sum^n_{i=2}\sum^{i-1}_{j=1} w_{ij} u_{ij} \right )$ then 

\begin{equation*}
    \hat{\beta}_1 - \hat{\beta}_{1,(i,j)} (w_{ij}) = \epsilon_n - \delta_n(w_{ij}). 
\end{equation*}

\noindent where $\hat{\beta}_{1(i,j)} (w_{ij})$ denotes an EwPO estimator with the weight $w_{ij}$. It is clear that under the null of exogeneity, $\epsilon_n=o_p(1)$ and $\delta_n(w_{ij}) = o_p(1)$ with the appropriate choice of $w_{ij}$. Under the alternative that $u_i$ and $x_i$ are correlated however, $|\epsilon_n - \delta_n(w_{ij})| = C + o_p(1)$ where $C$ is a positive (non-zero) real number. Note that if EwPO is consistent under the null and the alternative then $\delta_n(w_{ij}) = o_p(1)$ but $\epsilon_n$ does not vanish asymptotically. This was the basis of the Hausman test, but $\delta_n(w_{ij}) = o_p(1)$ is not required for the test to be valid since the test statistics still has non-zero expectation under the alternative if the asymptotic bias are different between the two estimators. This means as long as both estimators are consistent under the null and their asymptotic bias are different under the alternative, it is still possible to construct a Hausman-type test even though both estimators are inconsistent under the alternative. \par

\indent To demonstrate the general idea, consider the EwPO estimator with $w_{ij} = |\Delta x_{ij}|$, $\hat{\beta}_{1,(i,j)}$, Equation (\ref{eq:tempmess01}) implies that 

\begin{equation*}
    \hat{\beta}_1 - \hat{\beta}_{1,(i,j)} = \frac{\sum^n_{i=1} (x_i-\bar{x})u_i}{\sum^n_{i=1} (x_i-\bar{x})^2} - \frac{\sum^n_{i=2} \sum^{i-1}_{j=1} \sgn{\Delta x_{ij}}\Delta u_{ij}}{\sum^n_{i=2}\sum^{i-1}_{j=1} |\Delta x_{ij}|}. 
\end{equation*}

\noindent In the case that $u_i = \rho (x_i - \mu_x) + v_i$ with $x_i \perp v_i$, it is straightforward to show that $\epsilon_n - \delta_n(|\Delta x_{ij}|) = o_p(1)$ since both bias terms are $\rho + o_p(1)$. Thus, the test will not work with $\hat{\beta}_{1,(i,j)}$ in this case, which rule out the case when $x_i$ and $u_i$ follows a joint normal distribution.  \par

\indent However, when $u_i$ cannot be expressed as a linear function of $x_i$ then Hausman type test may be possible. One example is to consider $u_i = \rho x_i^3 + v_i$ where $x_i$ follows a symmetrical distribution with $\E(x_i) = 0$, $\E(x_i^4) = \kappa < \infty$ and $x_i \perp v_i$, then it is possible to show that $\epsilon_n = \kappa\rho + o_p(1)$ and $\delta_n (|\Delta x_{ij}|) = o_p(1)$ \footnote{EwPO with $w_{ij} = |\Delta x_{ij}|$ is consistent in this special case, despite correlation between $u_i$ and $x_i$}.  

\indent One possible way to construct the test statistics is to consider

\begin{equation}
    H(w_{ij}) = \left [ \hat{\bfbeta}_1 - \hat{\bfbeta}_{1,(i,j)}(w_{ij}) \right ]'\left \{ \Delta \bfSigma \left [\hat{\bfbeta}_1, \hat{\bfbeta}_{1,(i,j)} (w_{ij}) \right ] \right \}^{-1} \left [\hat{\bfbeta}_1 - \hat{\bfbeta}_{1,(i,j)} (w_{ij}) \right ] \,,
\end{equation}

\noindent where $\Delta \bfSigma \left [ \hat{\bfbeta}_1 - \hat{\bfbeta}_{1,(i,j)} \right ] = \bfSigma(\hat{\bfbeta}_1) - \bfSigma\left(\hat{\bfbeta}_{1,(i,j)}\right )$ with $\bfSigma(\hat{\bfbeta}_1)$ and $\bfSigma\left(\hat{\bfbeta}_{1,(i,j)} \right )$ denoting the variance-covariance matrices of $\hat{\bfbeta}_1$ and $\hat{\bfbeta}_{1, (i,j)}$, respectively. 

\indent The asymptotic distribution of the test statistics in the case $K=1$ with $w_{ij} = |\Delta x_{ij}|$ can be readily obtained using Proposition \ref{prop:abs_diffx_dist} as demonstrated in Proposition \ref{prop:hausman}. 

\begin{proposition} \label{prop:hausman}
 Under Assumptions \ref{ass:moment_x}, \ref{ass:moment_u} and \ref{ass:mixingale} along with the assumption that $|\epsilon_n - \delta_n(|\Delta x_{ij}|)| = C + o_p(1)$ where $C>0$, then  

\begin{equation} \label{eq:dist_hausman}
    H(|\Delta x_{ij}|) \overset{d}{\sim} \frac{\sigma_u^2}{d\mu^2_{\Delta x_s}} \left ( \frac{\mu_{\Delta x_s} - 2\sigma_x}{\sigma_x} W(1) - 4 \int^1_0 W(\lambda) d\lambda \right )^2
\end{equation}

\noindent where $d = Var(\hat{\beta}_1) - Var(\hat{\beta}_{1,(i,j)})$. 
\end{proposition}

\begin{proof}
    See Appendix \ref{app:proof}. 
\end{proof}

\indent The result from Proposition \ref{prop:hausman} is that the distribution of the test is non-standard and therefore its critical values must be simulated. Specifically, there are two components that would require simulation. First, the expression within the bracket of Equation (\ref{eq:dist_hausman}) and second, $Var \left ( \hat{\beta}_{1,(i,j)} \right )$, following the result from Proposition \ref{prop:abs_diffx_dist}. In both cases, Equations (\ref{eq:stochfirst}) to (\ref{eq:stochlast}) can be used to carry out the simulation. Other quantities, such as $\sigma_u$, $\sigma_x$ and $\mu_{\Delta x_S}$, can be replaced by their sample estimates. \par
\indent Table \ref{tab:cv_hausman} provides simulated critical values for the case $\sigma_u = 1$, $\sigma_x=1$ and $\mu_{\Delta x_S} = \sqrt{2/\pi}$ with $n=5000$ and 10,000 replications.  \par  
\begin{table}[h]
    \begin{center}
    \begin{tabular}{c|c}
        \toprule
        $\alpha$ & Critical Value \\ \midrule
        0.01\% & 2.5675 \\
        0.05\% & 2.5067 \\
        0.10\% & 2.5014 \\
       \bottomrule 
    \end{tabular}
       \caption{Simulated Critical Values of the EwPO Hausman-type Endogeneity test with $w_{ij} = |\Delta x_{ij}|$. } \label{tab:cv_hausman}
    \end{center}
\end{table}
\indent The power of the test depends on the exact relation between $x_i$ and $u_i$ as well as the choice of $w_{ij}$ and is an interesting issue for further research.

\section{Conclusion}
Endogeneity has been an enduring problem in econometrics. The practice to compensate for the information loss due to the correlation between the explanatory variable(s) and the disturbance terms has been to use additional information in the form of instrumental variables or moment conditions for testing and consistent estimation. The problem with this approach is that the results hinge on the “quality” of this information. The common wisdom to the day is that testing for endogeneity is not possible based exclusively on the information provided by the observations of the variables a model. This paper through the introduction of a new estimation method, the so-called estimation with pairwise observations (EwPO), demonstrates that this belief is misconceived and introduces three different direct endogeneity tests, ready to be used in applications. The EwPO new approach may also open the door for further new research paths in econometric and statistical estimation theory.

\vskip 3cm
\Appendix{Monte Carlo Simulations of EwPO Estimation with Selected Weights} \label{app:MC-estimation}
{

This Appendix presents some Monte Carlo (MC) simulation results to examine the properties of some EwPO estimators with selected weights. Additional results can be found in the Online Supplement (see \cite{EwPO-Sup2023}). \par

The data generating process for the MC simulations is based on the model

\begin{equation*}
    y_i = \beta_0 + \beta_1x_i + u_i
\end{equation*}

\noindent and the MC experiments consider two possible distributions for $u_i$, namely 

\begin{enumerate}
    \item \(u_{i} \sim N(0,1)\),
    \item \(u_{i} \sim \) skewed normal distribution,
\end{enumerate}
where the skewed normal distribution is generated as
\begin{equation*}
    u_i = \xi + \lambda |v_i| + z_i,
\end{equation*}
\noindent with $\xi=-\lambda\sqrt{\frac{2}{\pi}}$, $v_i\sim N(0,1)$ and $z_i\sim N(0, \sigma^2)$ such that $v_i$ and $w_i$ are independently distributed. \par

The MC experiments consider uniform distribution U(-10,10) for the regressor $x_i$. 


Two parameter vectors have been considered $(\beta_0, \beta_1) = (1,0.5)$ and $(\beta_0, \beta_1) = (1, 1.5)$ for purposes of robustness checking. The number of MC replications was 1000. For the sake of brevity, only the results for $(\beta_0, \beta_1) = (1,0.5)$ are reported here. The other results can be found in the Online Supplement  (\cite{EwPO-Sup2023}).








\begin{table}[H]
\centering
\begin{tabular}{|c|c|c|c|c|}
\hline
\multicolumn{1}{|c|}{Sorted MC} & \multicolumn{4}{c|}{Estimates and MC standard errors}                                             \\ \hline
\multirow{4}{*}{{n=50}}   & Parameter & Estimate/S.e. & OLS  & pairwise\\
\cline{2-5}  & 
\multirow{2}{*}{\(\hat{\beta_{0}}\)} & Estimate & 0.9866 & 0.9866\\ \cline{3-5} 
 & & S.e. & 0.1609&0.1609\\ \cline{2-5} 
& \multirow{2}{*}{\(\hat{\beta_{1}}\)} & Estimate & 0.4999&0.4999\\ \cline{3-5} 
 &  & S.e. &0.2928&0.2928\\ \hline
\multirow{4}{*}{{n = 500}}  & \multirow{2}{*}{\(\hat{\beta_{0}}\)}            & Estimate & 1.0001&1.0001\\ \cline{3-5} 
&   & S.e. & 0.0504&0.0504\\ \cline{2-5} 
& \multirow{2}{*}{\(\hat{\beta_{1}}\)}
& Estimate & 0.4971&0.4971\\ \cline{3-5} 
& & S.e. & 0.0948&0.0948\\ \hline
\multirow{4}{*}{{n = 5000}}   & \multirow{2}{*}{\(\hat{\beta_{0}}\)}            & Estimate & 0.9992&0.9992\\ \cline{3-5} 
&  & S.e. & 0.0169&0.0169\\ \cline{2-5} 
& \multirow{2}{*}{\(\hat{\beta_{1}}\)} & Estimate & 0.4997&0.4997\\\cline{3-5} 
&  & S.e. & 0.0281&0.0281\\ \hline
\end{tabular}
\caption{Sorted -- full-pairwise MC, $x_i$ $\sim U(-10,10)$, $u_i \sim $ skewed normal distribution, 
 $\Delta x$ weighted estimator}
\label{table:1}
\end{table}

\begin{table}[H]
\centering
\begin{tabular}{|c|c|c|c|c|}
\hline
\multicolumn{1}{|c|}{Sorted MC} & \multicolumn{4}{c|}{Estimates and MC standard errors}                                             \\ \hline
\multirow{4}{*}{{n=50}}   & Parameter & Estimate/S.e. & OLS  & pairwise\\
\cline{2-5}  & 
\multirow{2}{*}{\(\hat{\beta_{0}}\)} & Estimate & 0.9866&0.9866\\ \cline{3-5} 
 & & S.e. & 0.1609&0.1609\\ \cline{2-5} 
& \multirow{2}{*}{\(\hat{\beta_{1}}\)} & Estimate & 0.4999&0.4999\\ \cline{3-5} 
 &  & S.e. & 0.2928&0.2928\\ \hline
\multirow{4}{*}{{n = 500}}  & \multirow{2}{*}{\(\hat{\beta_{0}}\)}            & Estimate & 1.0001&1.0001\\ \cline{3-5} 
&   & S.e. & 0.0504&0.0504\\ \cline{2-5} 
& \multirow{2}{*}{\(\hat{\beta_{1}}\)}
& Estimate & 0.4971&0.4971\\ \cline{3-5} 
& & S.e. & 0.0948&0.0948\\\hline
\multirow{4}{*}{{n = 5000}}   & \multirow{2}{*}{\(\hat{\beta_{0}}\)}            & Estimate & 0.9992&0.9992\\ \cline{3-5} 
&  & S.e. & 0.0169&0.0169\\ \cline{2-5} 
& \multirow{2}{*}{\(\hat{\beta_{1}}\)} & Estimate & 0.4997&0.4997\\ \cline{3-5} 
&  & S.e. & 0.0281&0.0281\\ \hline
\end{tabular}
\caption{Non-sorted full-pairwise MC, $x_i$ $\sim U(-10,10)$, $u_i \sim $ skewed normal distribution, 
 $\Delta x$ weighted estimator}
\label{table:2}
\end{table}
Note: No mistake, the sorted and non-sorted full-pairwise estimation results given in Tables \ref{table:1} and \ref{table:2} are exactly the same. 

\begin{table}[H]
\centering
\begin{tabular}{|c|c|c|c|c|}
\hline
\multicolumn{1}{|c|}{Non-sorted MC} & \multicolumn{4}{c|}{Estimates and MC standard errors}                                             \\ \hline
\multirow{4}{*}{{n=50}}   & Parameter & Estimate/S.e. & OLS  & pairwise\\
\cline{2-5}  & 
\multirow{2}{*}{\(\hat{\beta_{0}}\)} & Estimate & 1.0009 & 0.9967 \\ \cline{3-5} 
 & & S.e. & 0.1408 & 0.1997 \\ \cline{2-5} 
& \multirow{2}{*}{\(\hat{\beta_{1}}\)} & Estimate & 0.5007 & 0.5018 \\ \cline{3-5} 
 &  & S.e. & 0.0251 & 0.0291 \\ \hline
\multirow{4}{*}{{n = 500}}  & \multirow{2}{*}{\(\hat{\beta_{0}}\)}            & Estimate & 0.9967 & 0.9966 \\ \cline{3-5} 
&   & S.e. & 0.0446 & 0.0640 \\ \cline{2-5} 
& \multirow{2}{*}{\(\hat{\beta_{1}}\)}
& Estimate & 0.5001 & 0.4998 \\ \cline{3-5} 
& & S.e. & 0.0081 & 0.0095 \\ \hline
\multirow{4}{*}{{n = 5000}}   & \multirow{2}{*}{\(\hat{\beta_{0}}\)}            & Estimate & 1.0001 & 1.0004 \\ \cline{3-5} 
&  & S.e. & 0.0144  & 0.0192 \\ \cline{2-5} 
& \multirow{2}{*}{\(\hat{\beta_{1}}\)} & Estimate & 0.4999 & 0.4999 \\ \cline{3-5} 
&  & S.e. & 0.0025 & 0.0031 \\ \hline
\end{tabular}
\caption{Non-sorted adjacent MC, $x_i \sim U(-10,10)$, $ u_i \sim N(0,1)$,
 $|\Delta x|$ weighted estimator}
\label{table:3}
\end{table}

\begin{table}[H]
\centering
\begin{tabular}{|c|c|c|c|c|}
\hline
\multicolumn{1}{|c|}{Full-pairwise MC} & \multicolumn{4}{c|}{Estimates and MC standard errors}                                             \\ \hline
\multirow{4}{*}{{n=50}}   & Parameter & Estimate/S.e. & OLS  & pairwise\\
\cline{2-5}  & 
\multirow{2}{*}{\(\hat{\beta_{0}}\)} & Estimate & 1.0009 & 1.0009 \\ \cline{3-5} 
 & & S.e. & 0.1408 & 0.1408 \\ \cline{2-5} 
& \multirow{2}{*}{\(\hat{\beta_{1}}\)} & Estimate & 0.5007 & 0.5007 \\ \cline{3-5} 
 &  & S.e. & 0.0251 & 0.0251 \\ \hline
\multirow{4}{*}{{n = 500}}  & \multirow{2}{*}{\(\hat{\beta_{0}}\)}            & Estimate & 0.9967 & 0.9967 \\ \cline{3-5} 
&   & S.e. &  0.0446 & 0.0446 \\ \cline{2-5} 
& \multirow{2}{*}{\(\hat{\beta_{1}}\)}
& Estimate & 0.5001 & 0.5001 \\ \cline{3-5} 
& & S.e. & 0.0081 & 0.0081 \\ \hline
\multirow{4}{*}{{n = 5000}}   & \multirow{2}{*}{\(\hat{\beta_{0}}\)}            & Estimate & 1.0001 & 1.0001 \\ \cline{3-5} 
&  & S.e. & 0.0144 & 0.0144 \\ \cline{2-5} 
& \multirow{2}{*}{\(\hat{\beta_{1}}\)} & Estimate & 0.4999 & 0.4999 \\ \cline{3-5} 
&  & S.e. & 0.0025 & 0.0025 \\ \hline
\end{tabular}
\caption{Sorted full-pairwise MC,  $ x_i \sim U(-10,10)$, $ u_i \sim N(0,1)$,
 $\Delta x$ weighted estimator}
\label{table:4}
\end{table}




\Appendix{Monte Carlo Simulations for the Covariance Test} \label{app:MC-test}

The MC setup considers again sample sizes $n=50$, $500$ and $5000$ with $1000$ replications. 

\begin{enumerate}[{Step} 1.]
\item Generated the same model as above with one explanatory variable:
$$
y_i = \beta_0 + \beta_1x_i + u_i
$$
with $\beta_0=1$, $\beta_1=0.5$ to star with, and $u_i$ was generated as
$N(0, 1)$. Finally, $x$ was generated as $N(0, 5)$ and also as $U(-5, 5)$. \par

\indent The simulations of $x_i$ and $u_i$ were carried out under four different correlations, namely $\rho=0$ (benchmark exogeneity), $\rho=0.2$ (small), $\rho=0.5$ (medium), and $\rho=0.8$ (large).

\item Estimate the model with EwPO with $w_{ij} = \Delta x_{ij}$ and $w_{ij} =|\Delta x_{ij}|$. In each case, calculate the test statistics as defined in Equation (\ref{eq:testStatistics}). 

\end{enumerate}

For further simulation results please refer to the Online Supplement (\cite{EwPO-Sup2023}).

\begin{table}[H]
\centering
\begin{tabular}{|c|c|c|c|c|c|}
\hline
\multicolumn{1}{|c|}{Full-pairwise MC} & \multicolumn{5}{c|}{Average test statistics}                                             \\ \hline
\multirow{4}{*}{{n=50}}  & Parameter & Pairwise & \makecell{Standard\\deviation} & Skewness & Kurtosis\\
\cline{2-6}  & 
Exogen &  -0.0003  & 0.7842 & -0.1804 & 3.1856\\ 
\cline{3-6} 
& $\rho=0.2$&  -0.5478  & 0.8647 & 0.0273 & 3.1118 \\ 
\cline{3-6} 
& $\rho=0.5$ &  -0.6175  & 0.5483 & 0.0593 & 3.0018\\ 
\cline{3-6} 
& $\rho=0.8$&  -1.2625  & 0.4008 & 0.0122 & 3.1264\\ 
\hline

\multirow{4}{*}{{n = 500}} & \multirow{1}{*}{Exogen} & -0.0058   & 0.2175 & -0.0497 & 3.0011 \\ \cline{3-6} 
& $\rho=0.2$ &  -0.4241  & 0.2075 & -0.1293 & 3.0698 \\ 
\cline{3-6} 
& $\rho=0.5$ &  -1.1112  & 0.2017 & 0.0756 & 3.0814 \\ 
\cline{3-6} 
& $\rho=0.8$ &  -1.5581 & 0.1311 & -0.0891 & 2.8930\\ 
\hline

\multirow{4}{*}{{n = 5000}}   & Exogen  & -0.0031  & 0.0666 & 0.0101 & 2.8834 \\ \cline{3-6} 
& $\rho=0.2$ & -0.3995  & 0.0684 & 0.0123 & 3.1566\\ \cline{3-6} 
& $\rho=0.5$ &   -0.9737  & 0.0604 & -0.0610 & 2.9180 \\ \cline{3-6} 
& $\rho=0.8$ & -1.5681  & 0.0421 & 0.0637 & 3.0883\\ \cline{3-6} 
\hline
\end{tabular}

\caption{Average test statistics, full-pairwise MC,  $x_i \sim N(0,5)$, $\Delta x$ weighted estimator}
\label{table:n_abs1}
\end{table}

\begin{figure}[ht!]
\centering	
\includegraphics[scale=0.675]{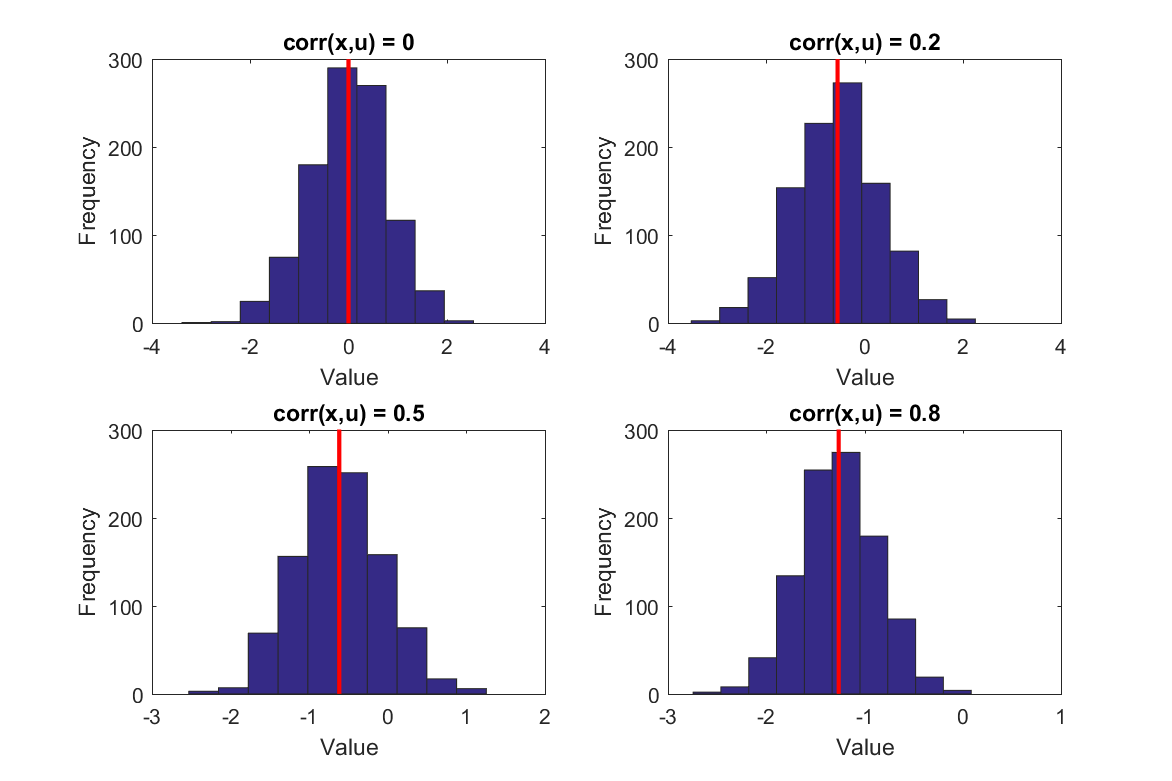}
		
\caption{Distribution of the test-statistics,  $x_i \sim$ N(0,5), $\Delta x$ weighted full-pairwise estimator, $n=50$}
\label{TestStatDist_Norm_noabs_50}
\end{figure}
\begin{figure}[bt!]
\centering	
\includegraphics[scale=0.675]{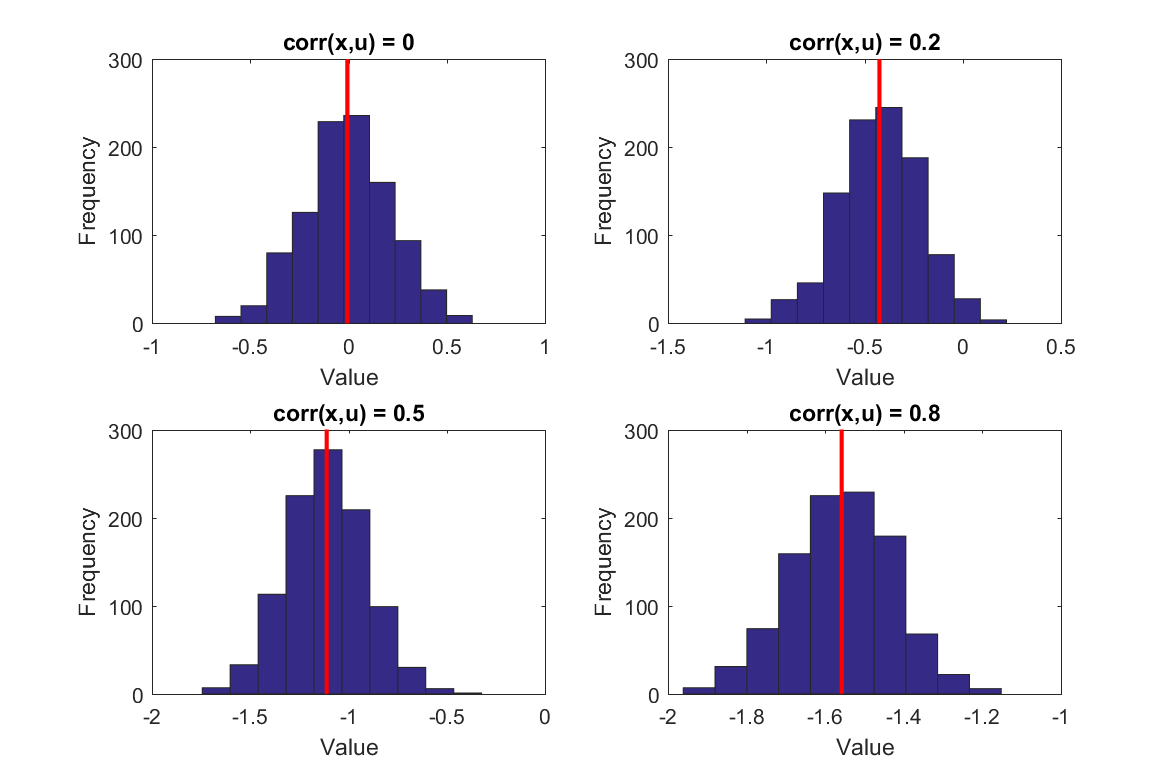}
		
\caption{Distribution of the test-statistics,  $ x_i \sim$ N(0,5), $\Delta x$ weighted full-pairwise estimator, $n=500$}
\label{TestStatDist_Norm_noabs_500}
\end{figure}

\begin{figure}[ht]
\centering	
\includegraphics[scale=0.675]{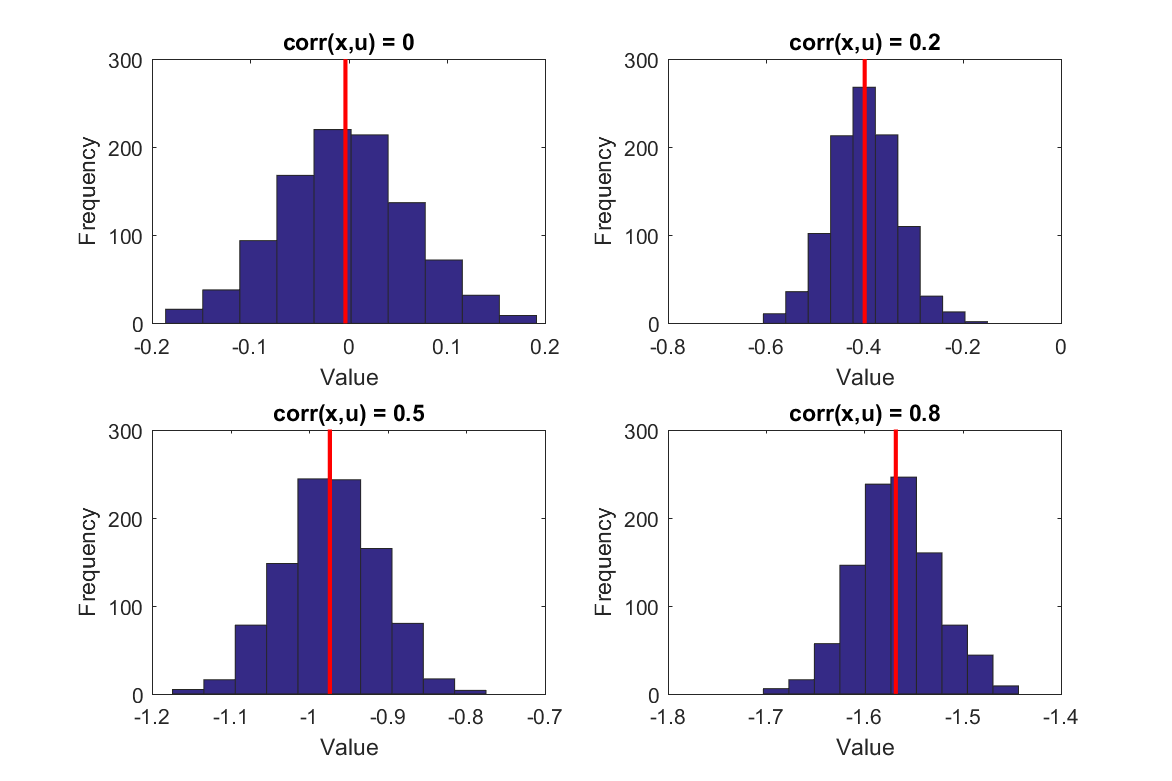}
		
\caption{Distribution of the test-statistics,  $x_i \sim$ N(0,5), $\Delta x$ weighted full-pairwise estimator, $n=5000$}
\label{TestStatDist_Norm_noabs_5000}
\end{figure}

\clearpage
\begin{table}[H]
\centering
\begin{tabular}{|c|c|c|c|c|c|}
\hline
\multicolumn{1}{|c|}{\makecell{Full-pairwise MC}} & \multicolumn{5}{c|}{Average test statistics}                                             \\ \hline
\multirow{4}{*}{{n=50}}  & Parameter & Pairwise & \makecell{Standard\\deviation} & Skewness & Kurtosis\\
\cline{2-6}  & 
Exogen &  0.0060  & 0.8890 & 0.0382 & 3.2387 \\ 
\cline{3-6} 
& $\rho=0.2$& -0.9247   & 0.7920 & -0.0244 & 2.9620\\ 
\cline{3-6} 
& $\rho=0.5$ &  -2.6616  & 0.7437 & 0.0121 & 3.2542\\ 
\cline{3-6} 
& $\rho=0.8$& -5.9118  & 0.5729 & -0.0305 & 3.0115\\ 
\hline

\multirow{4}{*}{{n = 500}} & \multirow{1}{*}{Exogen}  & -0.0234  & 0.2750 & -0.1429 & 3.1852\\ \cline{3-6} 
& $\rho=0.2$ & -1.0846   & 0.2471 & 0.0120 & 2.7804\\ 
\cline{3-6} 
& $\rho=0.5$ & -2.9665  & 0.2309 & 0.1625 & 3.2859 \\ 
\cline{3-6} 
& $\rho=0.8$ &  -4.2625   & 0.1494 & 0.0345 & 3.1253\\ 
\hline

\multirow{4}{*}{{n = 5000}}   & Exogen  & -0.0036 & 0.0785 & 0.0088 & 2.8812
  \\ \cline{3-6} 
& $\rho=0.2$ & -1.1328 & 0.0786 & -0.0625 & 3.1951
 \\ \cline{3-6} 
& $\rho=0.5$ &  -2.8321 & 0.0710 & -0.0394 & 2.9287
  \\ \cline{3-6} 
& $\rho=0.8$ &  -4.5415 & 0.0494 & 0.1022 & 3.0687
 \\ \cline{3-6} 
\hline
\end{tabular}


\caption{Average test statistics, full-pairwise MC,  $x_i \sim U(-5,5)$, $\Delta x$ weighted estimator}
\label{table:n_abs2}
\end{table}
\vskip 1cm
\begin{table}[H]
\centering
\begin{tabular}{|c|c|c|c|c|c|}
\hline
\multicolumn{1}{|c|}{Full-pairwise MC} & \multicolumn{5}{c|}{Average test statistics}                                             \\ \hline
\multirow{4}{*}{{n=50}}  & Parameter & Pairwise & \makecell{Standard\\deviation} & Skewness & Kurtosis\\
\cline{2-6}  &  Exogen  & 0.0191  & 0.8476 & -0.1264 & 3.1170\\
\cline{3-6} 
& $\rho=0.2$ & -1.2884  & 0.8520 & -0.1054 & 2.8218 \\ 
\cline{3-6} 
& $\rho=0.5$ & -2.8584   & 0.7488 & -0.0452 & 2.9483\\ 
\cline{3-6} 
& $\rho=0.8$  & -6.9515   & 0.6224 & -0.0069 & 3.2675\\ 
\hline

\multirow{4}{*}{{n = 500}} & \multirow{1}{*}{Exogen}  &  0.0010  & 0.2608 & -0.0092 & 2.9270
  \\ \cline{3-6} 
& $\rho=0.2$  & -1.2375  & 0.2588 & 0.0632 & 2.7668
 \\ \cline{3-6} 
& $\rho=0.5$  & -2.7598 & 0.2218 & -0.0266 & 2.9173
  \\ \cline{3-6} 
& $\rho=0.8$  & -4.3437  & 0.1479 & 0.0420 & 2.9936
 \\ 
 \hline

\multirow{4}{*}{{n = 5000}}   & Exogen  & -0.0003  & 0.0784 & 0.1196 & 2.8769
  \\ \cline{3-6} 
& $\rho=0.2$  & -1.1162  & 0.0790 & -0.1253 & 2.8993
 \\ \cline{3-6} 
& $\rho=0.5$  & -2.8175  & 0.0703 & -0.1483 & 3.0588
  \\ \cline{3-6} 
& $\rho=0.8$ & -4.5073  & 0.0490 & -0.0011 & 2.9966
 \\ \cline{3-6} 
\hline
\end{tabular}

\caption{Average test statistics, full-pairwise MC,  $x_i \sim U(-5,5)$, $|\Delta x|$ weighted estimator}
\label{table:n_abs3}
\end{table}

\pagebreak
\clearpage

\Appendix{Monte Carlo Results of the Residual Test and the Consistent Estimation} \label{app:residual}
The DGP in this case is 
\begin{equation*}
    y_i = 0.5x_i + u_i 
\end{equation*}
\noindent where $x_i \sim N(5,2)$ and $u_i \sim N(0,1)$. The MC experiments considered no correlation and $\rho = 0.2, 0.5$ and 0.8 with sample sizes ranging from 50 to 5000 for each case. 
\begin{table}[b!]
    \begin{center}
        \begin{tabular}{|c|c|c|c|c|c|}
        \hline
        Sample size&Correlation&Mean&Variance&Skewness&Kurtosis\\ \hline
        \multirow{4}{*}{n=50}&Exogen&0.0044&0.1604&-0.0176&0.1476\\ \cline{3-6}
        &$\rho = 0.2$&-0.4930&0.1546&-0.0546&0.0601\\ \cline{3-6}
        &$\rho =0.5$&-1.2506&0.1180&-0.0113&0.1834\\ \cline{3-6}
        &$\rho =0.8$&-2.0033&0.0606&0.0223&0.1587\\ 
        \hline
        \multirow{4}{*}{n=500}&Exogen&0.0032&0.0149&0.0839&0.0023\\ \cline{3-6}
        &$\rho =0.2$&-0.5005&0.0142&0.0264&-0.0169\\ \cline{3-6}
        &$\rho =0.5$&-1.2517&0.0118&0.023&0.1509\\ \cline{3-6}
        &$\rho =0.8$&-2.0015&0.0051&-0.0321&-0.2207\\
        \hline
        \multirow{4}{*}{n=1000} &Exogen&0.0009&0.0075&-0.0792&0.0818\\ \cline{3-6} 
        &$\rho =0.2$&-0.5027&0.0073&-0.0534&0.1680\\ \cline{3-6} 
        &$\rho =0.5$&-1.2474&0.0056&0.0687&0.0894\\ \cline{3-6}
        &$\rho =0.8$&-2.0003&0.0028&0.0432&-0.0320\\\cline{3-6} \hline
        \multirow{4}{*}{n=5000} & Exogen&0.0003&0.0015&-0.0341&0.1900\\\cline{3-6} 
        &$\rho =0.2$&-0.4993&0.0014&0.0541&0.1994\\ \cline{3-6}
        &$\rho =0.5$&-1.2507&0.0012&0.0573&0.1016\\ \cline{3-6}
        &$\rho =0.8$&-2.0003&0.0005&0.0722&-0.0429\\ 
        \hline
        \end{tabular}
        \caption{MC Results on the Distributions of $\hat{u}_i$.} \label{tab:meanuTest}
    \end{center}

\end{table}

\begin{table}[ht]
\begin{center}
    \begin{tabular}{|c|c|c|c|c|c|}
        \hline
        Sample size&Correlation&Mean&Variance&Skewness&Kurtosis\\ \hline
        \multirow{4}{*}{$n=50$} &Exogen&0.4993&0.0056&0.0584&0.2268\\ 
        &$\rho=0.2$&0.5984&0.0054&0.0641&0.067\\ 
        &$\rho=0.5$&0.7505&0.004&0.013&0.0595\\ 
        &$\rho=0.8$&0.9001&0.0021&-0.0344&0.3254\\ 
        \hline
        \multirow{4}{*}{$n=500$} &Exogen&0.4994&0.0005&-0.0126&0.0119\\ 
        &$\rho=0.2$&0.5998&0.0005&0.0287&-0.1062\\ 
        &$\rho=0.5$&0.7505&0.0004&-0.0401&0.0729\\ 
        &$\rho=0.8$&0.9003&0.0002&0.0101&-0.1041\\ 
        \hline
        \multirow{4}{*}{$n=1000$} &Exogen&0.4998&0.0003&0.0331&0.0474\\ 
        &$\rho=0.2$&0.6006&0.0003&0.0736&0.2013\\ 
        &$\rho=0.5$&0.7498&0.0002&-0.0369&-0.0212\\ 
        &$\rho=0.8$&0.9001&0.0001&-0.0594&-0.0057\\ \hline
        \multirow{4}{*}{$n=5000$} &Exogen&0.5000&0.0001&0.0240&0.0152\\ 
        &$\rho=0.2$&0.5999&4.877e-5&-0.0424&0.1421\\ 
        &$\rho=0.5$&0.7502&4.055e-5&-0.0478&0.2321\\ 
        &$\rho=0.8$&0.9000&1.850e-5&-0.0246&-0.0227\\ 
        \hline
    \end{tabular}
    \end{center}
    \caption{Monte Carlo Results on $\hat{\beta}$} \label{tab:biasbeta}
\end{table}

\begin{table}[H]
\begin{center}
\begin{tabular}{|c|c|c|c|c|c|}
    \hline
    Sample size&Correlation&Mean&Variance&Skewness&Kurtosis\\ \hline
    \multirow{4}{*}{$n=50$}&Exogen&0.5002&0.0008&0.0503&-0.0639\\ 
    &$\rho=0.2$&0.4998&0.0008&0.0861&0.2677\\ 
    &$\rho=0.5$&0.5004&0.0006&0.0015&-0.0670\\ 
    &$\rho=0.8$&0.4995&0.0003&0.0377&0.0951\\ 
    \hline
    \multirow{4}{*}{$n=500$} &Exogen&0.5000&8.031e-5&-0.0401&0.0146\\ 
    &$\rho=0.2$&0.4997&7.297e-5&-0.1208&0.0535\\ 
    &$\rho=0.5$&0.5002&6.249e-5&0.0534&0.0321\\ 
    &$\rho=0.8$&0.5000&2.946e-5&0.0825&0.1877\\ 
    \hline
    \multirow{4}{*}{$n=1000$}&Exogen&0.5000&3.975e-5&-0.0374&-0.1315\\ 
    &$\rho=0.2$&0.5000&3.871e-5&0.0496&-0.0301\\ 
    &$\rho=0.5$&0.5003&2.94e-5&-0.1190&0.0958\\ 
    &$\rho=0.8$&0.5001&1.435e-5&0.0450&0.1941\\ 
    \hline
    \multirow{4}{*}{$n=5000$}& Exogen&0.5000&8.225e-6&0.0703&-0.0280\\ 
        &$\rho=0.2$&0.5000&7.724e-6&0.0096&0.0661\\ 
        &$\rho=0.5$&0.5000&5.731e-6&-0.0170&-0.0944\\ 
        &$\rho=0.8$&0.5000&2.933e-6&0.0195&0.0369\\ 
    \hline
\end{tabular}
    \caption{Bias-corrected $\hat{\beta}$} \label{tab:biasedCorrected}
\end{center}
\end{table}

\vfill\eject
\Appendix{Mathematical Appendix}  \label{app:proof}
\noindent \textbf{Proof of Proposition \ref{prop:dist_x}}. \par
Define  
\begin{align}
    B_n(\lambda) =& n^{-1} \sum_{i=1}^{\floor{\lambda n}} \frac{ x_i - \mu_x }{\sigma_x} \\
    W_{n}(\lambda) =& n^{-1} \sum_{i=1}^{\floor{\lambda n}} \frac{u_i}{\sigma_u} \\
\end{align}
\noindent where $1/n \leq \lambda \leq 1$ and $\floor{x}$ denotes the largest integer contains in $x$. Under Assumptions (\ref{ass:moment_x}) to (\ref{ass:mixingale}) $B_n(\lambda) \overset{d}{\rightarrow} B(\lambda)$ and $W_n(\lambda) \overset{d}{\rightarrow} W(\lambda)$ following the arguments in \cite{mcleish_invariance_1975}.  \par
\indent The proof first identifies the asymptotic distribution of each term in Equation (\ref{eq:testStatistics}) and the result follow by an application of the Continuous Mapping Theorem as shown in \cite{billingsley:1999}. From Equation (\ref{eq:partialSum_DeltaXij}), it is straightforward to show that 
\begin{equation} \label{eq:prop1_c1}
    n^{-\frac{3}{2}} \sigma_x^{-1} \sum^n_{i=2}\sum^{i-1}_{j=1} \Delta x_{ij} \overset{d}{\rightarrow} B(1) - 2\int^1_0 B(\lambda) d\lambda 
\end{equation}
\noindent and follow the same argument, it can also be shown that
\begin{equation}\label{eq:prop1_c2}
    n^{-\frac{3}{2}} \sigma_u^{-1} \sum^n_{i=2}\sum^{i-1}_{j=1} \Delta u_{ij} \overset{d}{\rightarrow} W(1) - 2\int^1_0 W(\lambda) d\lambda. 
\end{equation} 
\noindent Now, note that $\Delta x_{pq}^2 = \left [ (x_p - \mu_x) - (x_q-\mu_x) \right ]^2$  and
\begin{align*}
    \sum^n_{i=2}\sum^{i-1}_{j=1} & \Delta x^2_{pq} = \sigma^2_x \sum^n_{i=2}\sum^{i-1}_{j=1} \left ( \frac{x_i - \mu_x}{\sigma_x} - \frac{x_j-\mu_x}{\sigma_x} \right )^2 \\
    =& \sigma^2_x  \left [ \sum^n_{i=2}\sum^{i-1}_{j=1} \left ( \frac{x_i - \mu_x}{\sigma_x} \right )^2 + \sum^n_{i=2}\sum^{i-1}_{j=1} \left ( \frac{x_j-\mu_x}{\sigma_x} \right )^2 -  \sum^n_{i=2}\sum^{i-1}_{j=1} \frac{ (x_i-\mu_x)(x_j-\mu_x)}{\sigma^2_x} \right ] \\
    =& \sigma^2_x n^2 \left ( B^2_n(1) - 2\int^1_0 B^2_n(\lambda) d\lambda \right ) 
\end{align*}
\noindent The last line follows from a similar argument as Equations (\ref{eq:prop1_c1}) and (\ref{eq:prop1_c2}). Therefore, 
\begin{equation} \label{eq:prop1_c3}
    n^{-2} \sum^n_{p=2}\sum^{p-1}_{q=1} \Delta x^2_{pq} \overset{d}{\rightarrow} \sigma^2_x \left ( B^2(1) - 2 \int^1_0 B^2(\lambda) d\lambda \right ).  
\end{equation}
\noindent For the last term, repeat the argument above and note that $u_p = W_n(\frac{p}{n})-W_n(\frac{p-1}{n})$ 
\begin{align*}
  &\sum^n_{p=2}\sum^{p-1}_{q=1} \Delta x_{pq}\Delta u_{pq} \\
  =&  \sum^n_{p=2}\sum^{p-1}_{q=1} x_pu_p + x_qu_q -x_pu_q -x_qu_p \\ 
  =& n^2 \sigma_x\sigma_u \left ( B_n(1)W_n(1) + \int^1_0 B_n(\lambda) W_n(\lambda) d\lambda - \int^1_0 W_n(\lambda)dB_n(\lambda) - \int^1_0 B_n(\lambda)dW_n(\lambda) \right ). 
\end{align*}
\noindent Therefore 
\begin{equation} \label{eq:prop1_c4}
\begin{split}
    &n^{-2}  \sum^n_{p=2}\sum^{p-1}_{q=1} \Delta x_{pq}\Delta u_{pq} \\
    &\overset{d}{\rightarrow} \sigma_x\sigma_u \left ( B(1)W(1) + \int^1_0 B(\lambda)W(\lambda) d\lambda -\int^1_0B(\lambda)dW(\lambda) - \int^1_0W(\lambda)dB(\lambda) \right ). 
\end{split}
\end{equation}
\noindent Combining Equations (\ref{eq:prop1_c1}), (\ref{eq:prop1_c2}), (\ref{eq:prop1_c3}) and (\ref{eq:prop1_c4}) and apply the Continuus Mapping Theorem gives the result. This completes the proof.  $\blacksquare$

\noindent \textbf{Proof of Proposition \ref{prop:abs_diffx_dist}}.  \par

The key to the proof is the fact that $\sum^n_{i=2} \sum^{i-1}_{j=1} |\Delta x_{ij}| = \sum^n_{i=2} \sum^{i-1}_{j=1} \Delta x_{s,ij}$. Under Assumption \ref{ass:exogeneity}, the result as stated in Equation (\ref{eq:prop1_c2}) still holds since reordering $x_i$ does not affect the distribution of $u_i$. Thus 

\begin{equation} \label{eq:prop3_c1}
\begin{split}
   & \frac{2}{n(n-1)} \sum^n_{i=1}\sum^{i-1}_{j=1} |\Delta x_{ij}| \\
   =&\frac{2}{n(n-1)} \sum^n_{i=1}\sum^{i-1}_{j=1} \Delta x_{s,ij} \\
    =& \mu_{|\Delta x_s|} + o_p(1) 
\end{split}
\end{equation}

\noindent by the Law of Large Number for triangular arrays. Recall

\begin{equation*}
    \beta - \hat{\beta} = \left ( \sum^n_{i=1}\sum^{i-1}_{j=1} |\Delta x_{ij}| \right )^{-1} \left ( \sum^n_{i=1}\sum^{i-1}_{j=1} \Delta u_{ij} \right )
\end{equation*}

\noindent Hence, combining Equations (\ref{eq:prop1_c2}) and (\ref{eq:prop3_c1}) and apply the Continuous Mapping Theorem gives the result. This completes the proof. $\blacksquare$

\noindent \textbf{Proof of Proposition \ref{prop:abs_diffx_exotest}} \par
The proof follows the same argument as Proposition \ref{prop:dist_x}. Proposition \ref{prop:abs_diffx_dist} shows that the rate of convergence for EwPO with $w_{ij} = |\Delta x_{ij}|$ is $\sqrt{n}$ and therefore, 

\begin{equation*}
    S(\bfw) = \sum^n_{p=1}\sum^{p-1}_{q=1} \Delta u_{pq} \Delta x_{pq} + o_p(1)
\end{equation*}
\noindent since EwPO is a consistent estimator under Assumption \ref{ass:exogeneity} i.e., $\Delta \hat{u}_{pq} = \Delta u_{pq} + o_p(1)$. Therefore the result follows directly from Equation (\ref{eq:prop1_c4}). This completes the proof. $\blacksquare$ \par

\noindent \textbf{Proof of Proposition \ref{prop:hausman}} \par
Using the results from Proposition \ref{prop:abs_diffx_dist}, the proof follows the same argument as Theorem 1 in \cite{phillips_understanding_1986}. This completes the proof.  $\blacksquare$

\vskip 3cm
\bibliographystyle{chicago} \bibliography{ref} 

\end{document}